\documentclass[journal,final,onecolumn]{IEEEtran}

\usepackage[dvips]{graphicx}
\usepackage{epsfig}
\usepackage[cmex10]{amsmath}
\usepackage{amssymb}
\usepackage{amsthm}
\usepackage{amsfonts}
\usepackage{bm}
\usepackage{cite}
\usepackage[svgnames]{xcolor} 
\usepackage{pstricks,pst-node,pst-plot,pstricks-add}
\usepackage{algpseudocode}
\usepackage{algorithm}
\usepackage{tikz}
\usepackage{url}
\usepackage{subcaption}

\graphicspath{{figs/}}

\input{mohammad.def}

\begin{document}

	\title{Private Shotgun DNA Sequencing}
	
	\author{Ali Gholami, Mohammad Ali Maddah-Ali, and Seyed Abolfazl Motahari%
		\thanks{Ali Gholami is with the Department of Electrical Engineering, Sharif University of Technology, Tehran, Iran. (email:ali.gholami@ee.sharif.edu). Mohammad Ali Maddah-Ali is with Nokia Bell Labs, Holmdel, New Jersey, USA. (email:mohammad.maddahali@nokia-bell-labs.com). Seyed Abolfazl Motahari is with the Department of Computer Engineering, Sharif University of Technology, Tehran, Iran. (email:motahari@sharif.edu)}%
	}

	\maketitle

	\begin{abstract}

	Current techniques in sequencing a genome allow a service provider (e.g. a sequencing company) to have full access to the genome information, and thus the privacy of individuals regarding their lifetime secret is violated. In this paper, we introduce the problem of private DNA sequencing, where the goal is to keep the DNA sequence private to the sequencer. We propose an architecture, where the task of reading fragments of DNA and the task of DNA assembly are separated, the former is done at the sequencer(s), and the later is completed at a local trusted data collector. To satisfy the privacy constraint at the sequencer and reconstruction condition at the data collector, we create an information gap between these two relying on two techniques: (i) we use more than one non-colluding sequencer, all reporting the read fragments to the single data collector, (ii)  adding the fragments of some known DNA molecules, which are still unknown to the sequencers, to the pool. We prove that these two techniques provide enough freedom to satisfy both conditions at the same time. 
	\end{abstract}


	\section{Introduction}
	\label{sec:introduction}
	
The human genome is a long string comprised of roughly 3 billion units called nucleotides. Each nucleotide is chosen from a set of four possible types $\{A,C,G,T\}$. Human genomes are more than 98 percents alike. Variations within a human population can be divided into common variants and rare variants. Most of the differences among individuals' genomes can be captured by the common variants known as Single Nucleotide Polymorphisms (SNPs). An individual's genome can be identified uniquely by these SNPs. The process of identifying the SNPs of an individual is called genotyping. 

Having access to the genome of an individual reveals critical information regarding their phenotype~\cite{lander2011initial}, \cite{Stan}. It can be used to predict the risk of a particular disease \cite{lander1996new}, \cite{risch1996future}, \cite{reich2001allelic} and to propose personal medical care. As a result, the rate of genetic testing services has risen dramatically in recent years \cite{grosse2006clinical,american2003american}. On the other hand,  this information can be abused, for example by insurance companies to increase the rates for individuals. Therefore, keeping the confidentiality of this information is a major concern~\cite{GenomeTheft}, \cite{heeney2011assessing}.  Adding to the concern, we should note that disclosure of this information can also put one's relatives at risk  due to inherited similarities \cite{knoppers2002genetic}. 
Privacy issues in genomics have been the focus of many research initiatives~\cite{kaye2012tension, lunshof2008genetic, malin2004not, lowrance2007identifiability}, mainly investigating anonymization techniques  in a DNA data set \cite{sweeney2002k, aggarwal2005k, johnson2013privacy, fienberg2011privacy, kaye2009data, erlich2014routes}. However, leaking DNA information in the process of sequencing has been completely overlooked.

In this paper, we aim  to propose and analyze  some technique with provable guarantees in keeping the genomic information private. In the proposed scheme, we still use existing sequencing machines for sequencing, however we prevent these machines from assembling the genomes. 
We note that the process of sequencing consists of two phases: (1) reading, i.e. identifying the sequence of genomes in each fragment of DNA. (2) processing, i.e. concatenation of the read fragments of DNA and assembling the entire sequence. The proposed solution in this paper is based on separating these two steps as follows. The reading step which is technological and expensive is outsourced, and the processing step which is computational is done locally, in-house, on trusted machines. We increase the ambiguity at the reading machines such that the processing step is impossible to be executed there. To increase the ambiguity, in a pooled sequencing framework, we add the fragments of some individuals with known DNA sequences to the pool. On top of that, we distribute the reading task among several sequencers and collect all the results to complete the processing task at a local processor, referred to as \textit{data collector}. This will increase the knowledge gap between the data collector and each sequencer and allows us to complete the processing step at the data collector while limiting the information leakage at each sequencer.

This problem is conceptually connected to the Private Information Retrieval (PIR) problem \cite{chor1995private} \cite{sun2017capacity}, \cite{banawan2018capacity}, \cite{gertner2000protecting}. In PIR, there is a database that we want to retrieve a particular information from without letting it know what we have retrieved. In \cite{chor1995private} it was shown that if we have just one database, we should download all the information in the database so that it can not distinguish which data we actually aimed to retrieve, which is a trivial solution. The interesting observation in PIR is that having multiple databases helps us to reduce the communication load significantly while preserving privacy. This motivates us to use multiple sequencers rather than one sequencer in the process of sequencing.  

The rest of the paper is organized as follows. The problem setting is provided in Section \ref{sec2}. In Section \ref{sec3}, an achievable scheme is introduced. In Sections \ref{sec:static channel} and \ref{sec:random channel}, the mathematical modelings underlying the scheme are introduced and the results are provided. Section \ref{conclusion} concludes the paper and introduces some future steps.

	\section{Problem Setting} \label{sec2}
	
	We propose a system architecture that includes a trusted data collector and a set of $S$ non-colluding sequencers, for some $S \in \mathbb{N}$. In addition, There is a set of $M\in \mathbb{N}$ individuals,  labeled as individual 1 to individual $M$, who are interested to know their genome sequences, without leaking their DNA information to any of the sequencers. The data collector gathers the genomes of these individuals and sends the fragments of a particular subset of these individuals to each sequencer. Each sequencer $s\in [S]$ reads these fragments (\emph{reading phase}), and sends the resulting \emph{reads} to the data collector. 
Then the data collector uses all the reads received from all sequencers and assembles the DNA sequences for each individual (\emph{processing phase}). One should note here that, unlike conventional approaches,  in the scenario proposed in this paper, we separate the reading phase, done in the sequencers,  and the processing phase, completed in the trusted in-house data collector,  in order to provide DNA privacy. Indeed, our objective is to design the parameters of the system such that the sequencers \emph{collectively} can provide enough reads for the data collector to be able to assemble the DNA sequences, while each sequencer \emph{alone} does not have enough information to reconstruct some parts of the sequences.  

 It is also assumed that there is a set of $K\in \mathbb{N}$ individuals,  that the data collector knows their sequences a--priori, while the sequencers have no information about their sequences. In the proposed scenario, the DNA sequences of these individuals are used to increase the ambiguity, even more, at the sequencers, as compared to the data collector and improve privacy, as follows. The data collector adds the DNA fragments of these $K$ individuals, to the set of fragments that it sends to each sequencer. We call these $K$ individuals with known DNA sequences as \emph{known individuals}, and those $M$ individuals with unknown DNA sequences as \emph{unknown individuals}.

%
%
%

%
	
The genome of each individual can be characterized by genotyping their SNPs. Note that in every SNP position, two out of four types of bases (or alleles) can occur. The set of the two possible bases occurring in every SNP position is a public information. In every SNP position, the base occurring with more frequency is called major allele and the other one is called minor allele. In haploid individuals, in which there is one copy of the genome, this is equivalent to encoding the genome by a sequence of $1$s and $-1$s that indicate major and minor alleles respectively. Thus, the DNA sequence (or equivalently SNP sequence) can be represented by a vector in  $\{-1,1\}^N$, where $N\in \mathbb{N}$ denotes the total number of SNPs in the desired region of sequencing. 
We define the matrix $\mathbf{X} \in \{-1,1\}^{M\times N}$, where the entry $X_{m,n} \in \{-1,1\}$ is a random variable, denoting  the allele of unknown individual $m \in [M]$ in the SNP position $n\in [N]$. Similarly, we define the matrix $\mathbf{Y} \in \{-1,1\}^{K\times N}$, where the entry $Y_{k,n} \in \{-1,1\}$ is a random variable, denoting  the allele of the known individual $k \in [K]$ in the SNP position $n\in [N]$.


The data collector sends a set of fragments $\bigcup\limits_{m=1}^{M}\bigcup\limits_{n=1}^{N} \mathcal{F}^{(s)} _{m,n}+\bigcup\limits_{k=1}^{K} \bigcup\limits_{n=1}^{N} \tilde{\mathcal{F}}^{(s)} _{k,n}$ to the sequencer $s \in [S]$, where $\mathcal{F}^{(s)} _{m,n}$ and $\tilde{\mathcal{F}}^{(s)} _{k,n}$ denote the set of fragments containing SNP position $n$ for unknown individual $m$ and known individual $k$, respectively (see Fig.~\ref{fig:boat1}).
 We define random variables $\alpha^{(s)}_{m,n} \triangleq |\mathcal{F}^{(s)} _{m,n}|$ and $\tilde{\alpha}^{(s)}_{k,n} \triangleq |\tilde{\mathcal{F}}^{(s)} _{k,n}|$ as the coverage depth for SNPs  $X_{m,n} $ and $Y_{k,n}$,  respectively, at sequencer $s$. As will be seen later on, as the coverage depth increases, there will be more resistance to the error produced by sequencers in the reading process.
Every sequencer then reads the received fragments, with error, and reports these reads to the data collector. The set of reads reported to the data collector by the sequencer $s$ is denoted by $\mathcal{R}^{(s)}$. 
	
	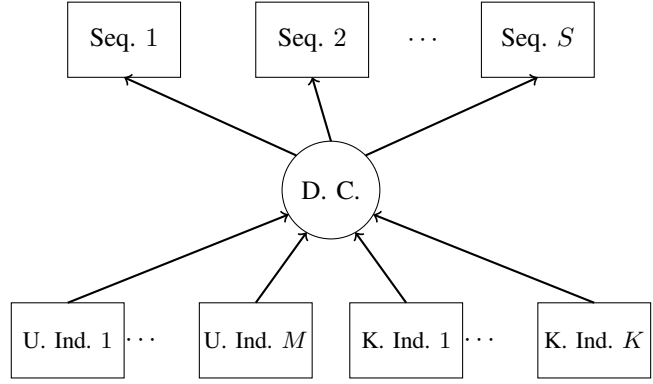
\begin{figure}
		\begin{center}
			\begin{tikzpicture}
			\draw (0,0) rectangle (1.5,1)node [pos=.5] {Seq. $1$};
			\draw (2.5,0) rectangle (4,1)node [pos=.5] {Seq. $2$};
			\node at (4.75,0.5) {$\cdots$};
			\draw (5.5,0) rectangle (7,1)node [pos=.5] {Seq. $S$};
			\draw (3.5,-1.5) circle (6.5 mm) node {D. C.};
			\draw (-0.75,-4) rectangle (0.75,-3)node [pos=.5] {\small U. Ind. $1$};
			\node at (1,-3.5) {$\cdots$};
			\draw (1.75,-4) rectangle (3.25,-3)node [pos=.5] {\small U. Ind. $M$};
			\draw (3.75,-4) rectangle (5.25,-3)node [pos=.5] {\small K. Ind. $1$};
			\node at (5.5,-3.5) {$\cdots$};
			\draw (6.25,-4) rectangle (7.75,-3)node [pos=.5] {\small K. Ind. $K$};
			\draw [thick, ->] (3.5,-8.5 mm) -- (3.25,0);
			\draw [thick, ->] (30.4 mm,-10.4 mm) -- (0.75,0);
			\draw [thick, ->] (39.6 mm,-10.4 mm) -- (6.25,0);
			\draw [thick, ->] (0,-3) -- (29.37 mm,-18.25 mm);
			\draw [thick, ->] (2.5,-3) -- (31.75 mm,-20.63 mm);
			\draw [thick, ->] (4.5,-3) -- (38.25 mm,-20.63 mm);
			\draw [thick, ->] (7,-3) -- (40.63 mm,-18.25 mm);
			\end{tikzpicture}
		\end{center}
		\caption{The block diagram of the proposed scheme in stage 1. First, the fragments of some individuals (known and unknown) are collected by the data collector, then they are pooled in a specific manner and sent to the sequencers.}
		\label{fig:boat1}
	\end{figure}
	
	\begin{figure}
		\begin{center}
			\begin{tikzpicture}
			\draw (0,0) rectangle (1.5,1)node [pos=.5] {Seq. $1$};
			\draw (2.5,0) rectangle (4,1)node [pos=.5] {Seq. $2$};
			\node at (4.75,0.5) {$\cdots$};
			\draw (5.5,0) rectangle (7,1)node [pos=.5] {Seq. $S$};
			\draw (3.5,-1.5) circle (6.5 mm) node {D. C.};
			\draw (0,-4) rectangle (1.5,-3)node [pos=.5] {\small U. Ind. $1$};
			\draw (2.5,-4) rectangle (4,-3)node [pos=.5] {\small U. Ind. $2$};
			\node at (4.75,-3.5) {$\cdots$};
			\draw (5.5,-4) rectangle (7,-3)node [pos=.5] {\small U. Ind. $M$};
			\draw [thick, ->] (3.25,0) -- (3.5,-8.5 mm);
			\draw [thick, ->] (0.75,0) -- (30.4 mm,-10.4 mm);
			\draw [thick, ->] (6.25,0) -- (39.6 mm,-10.4 mm);
			\draw [thick, ->] (3.5,-21.5 mm) -- (3.25,-3);
			\draw [thick, ->] (30.4 mm,-19.6 mm) -- (0.75,-3);
			\draw [thick, ->] (39.6 mm,-19.6 mm) -- (6.25,-3);
			\draw (5.5,-8.5 mm) rectangle (9,-21.5 mm)node [pos=.5]
			{
				\begin{tabular}{cc}
				Sequence information \\
				of known individuals \\
				\end{tabular}
			};
			\draw [thick, ->] (5.5,-1.5) -- (41.3 mm,-1.5);
			\end{tikzpicture}
		\end{center}
		\caption{In stage 2, each sequencer sends the results of the reads to the data collector.  Then using the information of the known individuals, the data collector will process the data and assemble the sequences of the unknown individuals.}
		\label{fig:boat2}
	\end{figure}
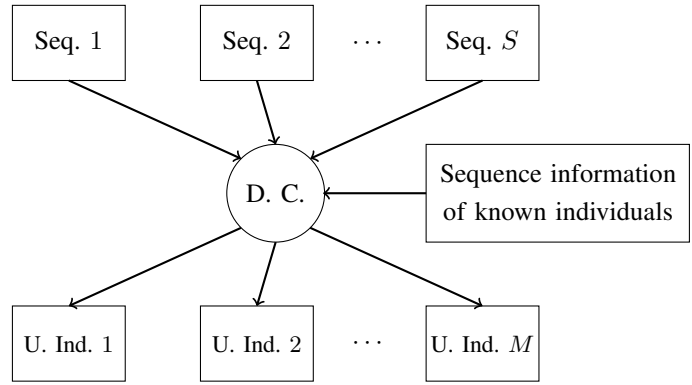
	
Sequencers have errors in reading bases. 
The probability of error in reading a SNP in a fragment is assumed to be constant across all sequences and for all SNPs and is denoted by  $\eta \in (0,1)$.
More precisely, in a sequencer, for a fragment of an individual, and in a SNP, the probability that a $1$ is read $-1$ or vice versa is $\eta$, independent of the sequencer, the individual,  the fragment, and the SNP.



Having $\mathbf{Y}$ as a side-information, the data collector maps $\mathcal{R}^{(s)}$, $\forall s\in [S]$ to the matrix $\hat{\mathbf{X}} \in \{-1, 1\}^{M \times N}$ 
using a function $\phi$, i.e.   
\begin{align*}
\hat{\mathbf{X}} = \phi \left(\mathbf{Y},  \mathcal{R}^{(s)},  s=1, \ldots ,S \right),
\end{align*}
where $\hat{\mathbf{X}}$ refers to an estimate of the matrix of SNPs for unknown individuals $\left(\mathbf{X}\right)$.


	
	
	The scenario should be such that the following two conditions are satisfied:
	\begin{itemize} \label{furmulbandi}
		\item \emph{Reconstruction Condition:}
		Let  $\mathbf{x}_n$ and $\hat{\mathbf{x}}_n$ denote the column $n$ of the matrix $\mathbf{X}$ and $\hat{\mathbf{X}}$ respectively. 
	 The reconstruction condition requires that the inequality below hold for any given  $\epsilon \in (0,1)$:
		\begin{equation} \label{rec}
		\mathbb{P}(\hat{\mathbf{x}}_n\ne \mathbf{x}_n)\leq \epsilon,\quad\forall n\in [N].
		\end{equation} 
				 $\epsilon$ is referred to as the \emph{accuracy level} and is a design parameter. 

		\item \emph{Privacy Condition:} 
	    Let $\mathcal{A}^{(s)}$ indicate the set of all unknown individuals who have fragments sent to sequencer $s\in[S]$ by the data collector.	
	    For privacy to be held,   for any sequencer $s\in [S]$, we want the distribution of $X_{m,n},   m \in \mathcal{A}^{(s)}, n \in [N]$ remains almost the same before and after reading the fragments. To be precise,  the privacy condition requires that the following inequality hold for any given  $\beta \in (0,1)$:
			\begin{equation} \label{pri}
		\frac{I\left(  X_{m,n},   m \in \mathcal{A}^{(s)}, n \in [N]   ;\mathcal{R}^{(s)}\right)}{|\mathcal{A}^{(s)}|N}\leq\beta.
		\end{equation}
		$\beta$ is referred to as the \emph{privacy level} and is a design parameter. 
	\end{itemize}
	
	In the next section,  we will introduce a proposed scheme that satisfies the two conditions simultaneously.

	\section{All-But-One Achievable Scheme} \label{sec3}
	In this section, we propose a particular solution that allows us to satisfy both the reconstruction~\eqref{rec} and privacy~\eqref{pri} conditions. We have two assumptions in our scheme:
	\begin{itemize}
		\item \emph{Assumption 1:} Every fragment is short enough to contain no more than one SNP.
		\item \emph{Assumption 2:} Every fragment is long enough that can be correctly mapped to the reference genome, i.e. we can identify exactly from what region of the genome's sequence it came from.
	\end{itemize}
	
	These two assumptions are realistic. We should keep in mind that there are approximately 3.3 million SNPs in the human genome. Comparing to the 3 billion length of the whole genome, it is concluded that the average distance between two SNPs is roughly 1000 base pairs \cite{shen2013comprehensive}. Moreover, using short read alignment algorithms like Bowtie \cite{langmead2009ultrafast}, it is possible to assemble reads of length in the order of a couple of hundreds. Thus using such algorithms, and choosing the fragments lengths to be about few hundreds, both assumptions are valid simultaneously.
	
	In the proposed achievable scheme, we focus on the case where $S=M$. In cases where $M$ is greater than $S$, we partition the set of individuals into groups of size $S$ and use this scheme for each group separately. We index unknown individuals from $1$ to $M$ ($M$ should be equal to or greater than $2$ is this scheme) and similarly known individuals from $1$ to $K$. In this paper, we propose a specific assignment scheme, to send the fragments of the individuals to the sequencers. In the proposed solution, named \emph{all-but-one scheme}, the data collector sends the fragments of all individuals, except individual $s$, to the sequencer $s$. More precisely, for $\forall s\in [S]$ and $\forall n\in [N]$
	\begin{align}
	\mathcal{F}_{s,n}^{(s)}=\varnothing,
	\end{align}
	or equivalently
	\begin{align} \label{www}
		\mathcal{A}^{(s)}=[S] \backslash \{s\}.
	\end{align}
	This means that for $\forall s\in [S]$ and $\forall n\in [N]$
	\begin{align}
	\alpha_{s,n}^{(s)}=0.
	\end{align}
	From~\eqref{www}, it is concluded that
	\begin{align} \label{mag}
	|\mathcal{A}^{(s)}|=S-1=M-1.
	\end{align}
	Also, to each sequencer, the fragments of all known individuals are sent by the data collector. In addition, we try to keep the coverage depths $\alpha_{m,n}^{(s)}$s and $\tilde{\alpha}_{k,n}^{(s)}$s almost constant and the same for all $s\in [S], m\in [M], k\in [K], n\in [N], m\neq s$. However, this is practically impossible. This means that $\alpha_{m,n}^{(s)}$s and $\tilde{\alpha}_{k,n}^{(s)}$s are random variables for all $s\in [S], m\in [M], k\in [K], n\in [N], m\neq s$. 
	Analyzing the problem with random $\alpha_{m,n}^{(s)}$s and $\tilde{\alpha}_{k,n}^{(s)}$ is rather complicated. 
	To make the analysis tractable, first, in Section \ref{sec:static channel}, we solve the problem for the case where these variables are constant and the same, and later in Section \ref{sec:random channel}, we generalize the results to the case where these variables are random.
		
	\section{All-But-One Scheme with Constant Coverage Depth}\label{sec:static channel}
	In this section for all $\forall n\in [N],\forall m\in [M],\forall k\in [K],\forall s\in [S], m\neq s$ we have
	\begin{align}
	\alpha^{(s)}_{m,n}=\tilde{\alpha}^{(s)}_{k,n}=\alpha, 
	\end{align}
	for some $\alpha \in \mathbb{N}$.

First, we introduce the main results. Then we derive the mathematical models in the data collector and the sequencers in Subsections \ref{mathdata} and \ref{mathdata2}, respectively.  We rely on these models to prove the main results in Subsections \ref{proof11} and \ref{proof12}. At last, we discuss the results in Subsection \ref{discuss1}.
	
	For the reconstruction condition to be satisfied, we have the following theorem.
	\begin{theorem}
		\label{thm1}
		In the all-but-one scheme with the constant coverage depth of $\alpha$ and reading probability of error of $\eta$, the reconstruction condition~\eqref{rec} is satisfied if
		\begin{equation} \label{recrec}
		\frac{\alpha}{M+K-1} \geq \frac{8\eta (1-\eta)}{(1-2\eta)^2} \ln\left(\frac{2^M}{\epsilon}\right).
		\end{equation}
	\end{theorem}
	
	For the privacy condition to be satisfied, we have the following theorem.
	\begin{theorem}\label{thm2}
		In the all-but-one scheme with constant coverage depth, the privacy constraint~\eqref{pri} is satisfied if
		\begin{align} \label{ineqthm3}
		f(K)\triangleq\frac{H\left(\text{Bin}\left(M-1+K,p_n\right)\right)-H(\text{Bin}\left(K,p_n\right))}{M-1}\leq\beta
		\end{align}
		holds for every $n\in [N]$, where $p_n$ is the major allele frequency for SNP position $n$,  $H(\cdot)$ refers to the discrete  entropy function, and $\text{Bin}(\cdot,\cdot)$ refers to the binomial random variable.
	\end{theorem}

	The main message of these results is that we can choose the parameters of the proposed scheme such that both conditions are satisfied, simultaneously.  In other words, these theorems confirm that the separation of the reading phase and the processing phase, together with using more than one sequencer and adding the sequences of known individuals to the pool, offers enough flexibility to satisfy both conditions at the same time; based on~\eqref{ineqthm3}, $K$ is chosen, and then using \eqref{recrec}, $\alpha$ is set. 
	
	 For more discussions on these results refer to~\ref{discuss1}.
	
	\subsection{Mathematical Model in Data Collector in the All-But-One Scheme} \label{mathdata}
	
	For any SNP position $n\in [N]$, the objective for the data collector is to estimate the vector $\mathbf{x}_n=\left[X_{1,n}, X_{2,n},\cdots, X_{M,n}\right]^T$.
	
	In the following, we will argue that in the all-but-one scenario, the knowledge provided by the sequencers to the data collector in SNP position $n$, can be modeled by a set of noisy linear equations as
	\begin{align}\label{modeleqn}
	\mathbf{g}_n=\mathbf{H}\mathbf{x}_n+\mathbf{z}_n,
	\end{align}
	where $\mathbf{H} = \mathbf{1}_{M\times M} - \mathbf{I}_{M\times M}$, and  $\mathbf{g}_n$ and $\mathbf{z}_n$ are $M\times 1$ vectors and $\mathbf{z}_n\sim\mathcal{N}(0,\sigma^2 \mathbf{I})$ in which
	\begin{align} \label{sigma2}
	\sigma^2\triangleq\frac{4(M+K-1)}{\alpha}\frac{\eta(1-\eta)}{(1-2\eta)^2}.
	\end{align}
In addition,  $\mathbf{z}_n$ and $\mathbf{x}_n$ are uncorrelated. From this set of noisy equations, the objective at the data collector is to reconstruct $\mathbf{x}_n$.
	
	For now, let us focus on the first sequencer. As mentioned before, fragments of unknown individuals indexed $2$ to $M$ are sent to sequencer $1$, as well as fragments of all known individuals. These fragments are pooled and have no tags regarding the index of individuals. So the information that the data collector receives is the number of fragments containing major (or minor) alleles. Note that the fragments are categorized by their corresponding SNP. Thus, in SNP position $n\in [N]$, it receives the result of the following sum
	\begin{align} \label{new}
	&\sum_{i=1}^{\alpha} \tilde{X}^{(1)}_{2,n,i} + \cdots +  \sum_{i=1}^{\alpha} \tilde{X}^{(1)}_{M,n,i}\nonumber\\
	+&\sum_{i=1}^{\alpha} \tilde{Y}^{(1)}_{1,n,i} + \cdots + \sum_{i=1}^{\alpha} \tilde{Y}^{(1)}_{K,n,i},
	\end{align}
	where $\tilde{X}^{(1)}_{m,n,i}$ and $\tilde{Y}^{(1)}_{k,n,i}$ are the noisy versions of $X_{m,n}$ and $Y_{k,n}$ after the reading phase in which there is an $\eta$ probability of error  in reading each fragment. Also $i$ refers to the index of fragments. As mentioned before, the data collector knows the sequence of known individuals a--priori. For known individual $k\in [K]$ and SNP position $n\in [N]$, we indicate the known alleles by $y_{k,n}\in \{-1,1\}$. Consequently for $\forall i\in [\alpha]$ we have
	\begin{align}
	\tilde{X}^{(1)}_{m,n,i} &=\left\{
	\begin{array}{ll}
	X_{m,n},   & \mathrm{w.p.}\quad 1-\eta \\‎ 
	-X_{m,n},  & \mathrm{w.p.}\quad \eta,
	\end{array}\right.\label{XXX}\\
	\tilde{Y}^{(1)}_{k,n,i} &=\left\{
	\begin{array}{ll}
	y_{k,n},   & \mathrm{w.p.}\quad 1-\eta \\‎ 
	-y_{k,n},  & \mathrm{w.p.}\quad \eta.
	\end{array}\right.
		\label{eq:what}
 	\end{align}
	
	After scaling~\eqref{new}, and subtracting  $\sum_{k=1}^{K} y_{k,n}$, we can assume that the data collector calculates $G^{(1)}_n$, defined as 
\begin{align}
	G^{(1)}_n \triangleq &\frac{1}{\alpha (1-2\eta)}\sum_{i=1}^{\alpha} \tilde{X}^{(1)}_{2,n,i} + \cdots +  \frac{1}{\alpha (1-2\eta)}\sum_{i=1}^{\alpha} \tilde{X}^{(1)}_{M,n,i}\nonumber\\
	+& \frac{1}{\alpha (1-2\eta)}\sum_{i=1}^{\alpha} \tilde{Y}^{(1)}_{1,n,i} + \cdots + \frac{1}{\alpha (1-2\eta)}\sum_{i=1}^{\alpha} \tilde{Y}^{(1)}_{K,n,i}\nonumber\\
	-&\sum_{k=1}^{K} y_{k,n}, \label{main1}
	\end{align}
Note that subtracting $\sum_{k=1}^{K} y_{k,n}$ in the above equation is fine, because the full knowledge of matrix $\mathbf{Y}$ is available at the data collector.

	In what follows, we calculate the mean and variance of $\tilde{X}^{(1)}_{m,n,i}$  conditioned on $X_{m,n}$
	\begin{align}
	\mathbb{E}\left(\tilde{X}^{(1)}_{m,n,i}\ |\ X_{m,n}\right)&=X_{m,n} (1-\eta)+(-X_{m,n})\eta,\nonumber\\
	&=(1-2\eta) X_{m,n}\\
	\mathrm{Var}\left(\tilde{X}^{(1)}_{m,n,i}\ |\ X_{m,n}\right)&=\mathbb{E}\left(\left(\tilde{X}^{(1)}_{m,n,i}\right)^2\ |\ X_{m,n}\right)\nonumber\\
	&-\left(\mathbb{E}\left(\tilde{X}^{(1)}_{m,n,i}\ |\ X_{m,n}\right)\right)^2\nonumber\\
	&=1-\left((1-2\eta) X_{m,n}\right)^2\nonumber\\
	&=4\eta (1-\eta),
	\end{align}
	in which the last inequality results from $X^2_{m,n}=1$. Then using the MMSE estimate and orthogonality principle, we can write
	\begin{align} \label{model1}
	\tilde{X}^{(1)}_{m,n,i}=(1-2\eta)X_{m,n}+Z^{(1)}_{m,n,i},
	\end{align} 
	where $Z^{(1)}_{m,n,i}$ is a random variable with $\mathbb{E}\left(Z^{(1)}_{m,n,i}\right)=0$ and $\textrm{Var}\left(Z^{(1)}_{m,n,i}\right)=4\eta (1-\eta)$. Also $Z^{(1)}_{m,n,i}$ and $X_{m,n}$ are uncorrelated.
	Consequently
	\begin{align} \label{new1}
	\frac{1}{\alpha (1-2\eta)}\sum_{i=1}^{\alpha} \tilde{X}^{(1)}_{m,n,i}=X_{m,n}+\frac{\sum_{i=1}^{\alpha}Z^{(1)}_{m,n,i}}{\alpha (1-2\eta)}.
	\end{align}
	Based on central limit theorem $\frac{\sum_{i=1}^{\alpha}Z^{(1)}_{m,n,i}}{\sqrt{\alpha}}$ converges in distribution to a normal distribution with zero mean and variance $4\eta(1-\eta)$. Thus
	\begin{align}
	\frac{\sum_{i=1}^{\alpha}Z^{(1)}_{m,n,i}}{\alpha (1-2\eta)}=\frac{1}{\sqrt{\alpha}(1-2\eta)}\frac{\sum_{i=1}^{\alpha}Z^{(1)}_{m,n,i}}{\sqrt{\alpha}}
	\end{align}
	converges in distribution to a normal distribution with zero mean and variance $\frac{4\eta(1-\eta)}{\alpha(1-2\eta)^2}$. Similarly
	\begin{align}
	\mathbb{E}\left(\tilde{Y}^{(1)}_{k,n,i}\ |\ y_{k,n}\right)&=y_{k,n} (1-\eta)+(-y_{k,n})\eta\nonumber\\
	&=(1-2\eta) y_{k,n}\\
	\mathrm{Var}\left(\tilde{Y}^{(1)}_{k,n,i}\ |\ y_{k,n}\right)&=\mathbb{E}\left(\left(\tilde{Y}^{(1)}_{k,n,i}\right)^2\ |\ y_{k,n}\right)\nonumber\\
	&-\left(\mathbb{E}\left(\tilde{Y}^{(1)}_{k,n,i}\ |\ y_{k,n}\right)\right)^2\nonumber\\
	&=1-\left((1-2\eta) y_{k,n}\right)^2\nonumber\\
	&=4\eta (1-\eta).
	\end{align}
	Therefore, we can write
	\begin{align} \label{model2}
	\tilde{Y}^{(1)}_{k,n,i}=(1-2\eta)y_{k,n}+\tilde{Z}^{(1)}_{k,n,i},
	\end{align}
	where $\tilde{Z}^{(1)}_{k,n,i}$ is a random variable with $\mathbb{E}\left(\tilde{Z}^{(1)}_{k,n,i}\right)=0$ and $\textrm{Var}\left(\tilde{Z}^{(1)}_{k,n,i}\right)=4\eta (1-\eta)$.
	Thus we have
	\begin{align} \label{new2}
	\frac{1}{\alpha (1-2\eta)}\sum_{i=1}^{\alpha} \tilde{Y}^{(1)}_{k,n,i}=y_k+\frac{\sum_{i=1}^{\alpha}\tilde{Z}^{(1)}_{k,n,i}}{\alpha (1-2\eta)},
	\end{align}
	Similar to $\frac{\sum_{i=1}^{\alpha}Z^{(1)}_{m,n,i}}{\alpha (1-2\eta)}$, $\frac{\sum_{i=1}^{\alpha}\tilde{Z}^{(1)}_{k,n,i}}{\alpha (1-2\eta)}$ converges in distribution to a normal distribution with zero mean and variance $\frac{4\eta(1-\eta)}{\alpha(1-2\eta)^2}$.
	
	Consequently, using~\eqref{new1} and~\eqref{new2} we can rewrite~\eqref{main1} as
	\begin{equation} \label{1111}
	G^{(1)}_n = X_{2,n} + \cdots + X_{M,n} +Z^{(1)}_{n},
	\end{equation}
	where $Z^{(1)}_{n}\sim\mathcal{N}\left(0,\sigma^2\right)$ in which $\sigma^2$ is defined in~\eqref{sigma2}.
		
	Using the same argument for all sequencers, we can see  that~\eqref{modeleqn} holds, for  $\mathbf{g}_n\triangleq\left[G^{(1)}_n,G^{(2)}_n,\cdots,G^{(M)}_n\right]^T$  and the noise vector $\mathbf{z}_n\triangleq\left[Z_n^{(1)},Z_n^{(2)}, \cdots, Z_n^{(M)}\right]^T$.
	
	\subsection{Mathematical Model in the Sequencers in the All-But-One Scheme} \label{mathdata2}
	
	Without loss of generality, again, consider sequencer $1$. It is important to note that the sequencer cannot differentiate between known and unknown individuals. In other words, all individuals are unknown for all sequencers. For the sequencer $1$, $\tilde{Y}^{(1)}_{k,n,i}$ for $\forall i\in[\alpha]$ can be written as 
	
	\begin{equation}
	\tilde{Y}^{(1)}_{k,n,i} =\left\{
	\begin{array}{ll}
	Y_{k,n},   & \text{w.p.}\quad 1-\eta \\‎ 
	-Y_{k,n},  & \text{w.p.}\quad \eta,
	\end{array}\right.
	\end{equation}
	 $\tilde{X}^{(1)}_{m,n,i}$ is also defined as~\eqref{XXX}. Sequencer 1 receives the summation in~\eqref{new}. Scaling this summation, Sequencer 1 receives $q^{(1)}_n$, defined as
	 \begin{align}
	 q^{(1)}_n \triangleq &\frac{1}{\alpha (1-2\eta)}\sum_{i=1}^{\alpha} \tilde{X}^{(1)}_{2,n,i} + \cdots +  \frac{1}{\alpha (1-2\eta)}\sum_{i=1}^{\alpha} \tilde{X}^{(1)}_{M,n,i}\nonumber\\
	 +& \frac{1}{\alpha (1-2\eta)}\sum_{i=1}^{\alpha} \tilde{Y}^{(1)}_{1,n,i} + \cdots + \frac{1}{\alpha (1-2\eta)}\sum_{i=1}^{\alpha} \tilde{Y}^{(1)}_{K,n,i},
	 \end{align}
	 Following the same steps as in the previous subsection, we can argue that $q^{(1)}_n$ can be written as
	\begin{align} \label{sequencer model}
	q^{(1)}_n& \triangleq X_{2,n}+X_{3,n}+\cdots +X_{M,n}\nonumber\\
	&+Y_{1,n}+Y_{2,n}+\cdots+Y_{K,n}+\tilde{Z}^{(1)}_n,
	\end{align} 
	where $\tilde{Z}^{(1)}_n\sim\mathcal{N}\left(0,\sigma^2\right)$ and $\sigma^2$ is determined in~\eqref{sigma2}.
	
	\subsection{Proof of Theorem \ref{thm1}} \label{proof11}
	Using the mathematical model introduced in~\eqref{modeleqn}, we are ready to provide the proof of Theorem \ref{thm1}.
	\begin{proof}
		Consider SNP position $n$. $\hat{\mathbf{x}}_n$ is the estimate of $\mathbf{x}_n$ by the data collector, based on minimum mean square error (MMSE) estimation. The probability of error is written as
		\begin{align}
		\mathbb{P}(\textrm{error}) = \mathbb{P}(\hat{\mathbf{x}}_n\neq\mathbf{x}_n)=
		\sum_{\mathbf{u}} \mathbb{P}(\mathbf{x}_n=\mathbf{u}) \mathbb{P}(\hat{\mathbf{x}}_n\neq\mathbf{u}\ |\ \mathbf{x}_n=\mathbf{u}).
		\end{align}
		Using union bound and substituting $\mathbb{P}(\hat{\mathbf{x}}_n\neq\mathbf{u}\ |\ \mathbf{x}_n=\mathbf{u})$ with $\mathbb{P}(\textrm{error}|\mathbf{u})$ we have
		\begin{align}
		\mathbb{P}(\textrm{error}|\mathbf{u}) &\leq \sum_{\mathbf{v}\neq \mathbf{u}} \mathbb{P}(\hat{\mathbf{x}}_n=\mathbf{v}\ |\ \mathbf{x}_n=\mathbf{u})\nonumber\\
		&=\sum_{\mathbf{v}\neq \mathbf{u}} \mathbb{P} \left(||\mathbf{g}_n-\mathbf{H}\mathbf{v}||_2 \leq ||\mathbf{g}_n-\mathbf{H}\mathbf{u}||_2 \ |\ \mathbf{x}_n=\mathbf{u}\right)\nonumber\\
		&=\sum_{\mathbf{v}\neq \mathbf{u}} \mathbb{P} (||\mathbf{H}(\mathbf{u}-\mathbf{v})+\mathbf{z}_n||_2 \leq ||\mathbf{z}_n||_2  \ |\ \mathbf{x}_n=\mathbf{u})\nonumber\\
		&=\sum_{\mathbf{v}\neq \mathbf{u}} \mathbb{P} \left((\mathbf{H}(\mathbf{u}-\mathbf{v}))\cdot\mathbf{z}_n \leq -\frac{||\mathbf{H}(\mathbf{u}-\mathbf{v})||_2^2}{2} \ |\ \mathbf{x}_n=\mathbf{u}\right)\nonumber\\
		&\stackrel{(a)}{=} \sum_{\mathbf{v}\neq \mathbf{u}} Q\left(\frac{||\mathbf{H}(\mathbf{u}-\mathbf{v})||_2}{2\sigma}\right)\nonumber\\
		&\stackrel{(b)}{\leq}2^M Q\left(\frac{1}{\sigma}\right)\nonumber\\
		&\stackrel{(c)}{=}2^M Q\left(\sqrt{\frac{\alpha (1-2\eta)^2}{4(M+K-1)\eta (1-\eta)}}\right),
		\label{eqn2}
		\end{align}
		where $(a)$ is due to the fact that
		\begin{align}
		(\mathbf{H}(\mathbf{u}-\mathbf{v}))\cdot\mathbf{z}_n\ \text{conditioned on}\ \mathbf{x}_n=\mathbf{u} \sim\mathcal{N}(0,||\mathbf{H}(\mathbf{u}-\mathbf{v})||_2^2\sigma^2),
		\end{align}
		and $(b)$ is due to the fact that the minimum norm  of $\mathbf{H}(\mathbf{u}-\mathbf{v})$ is $2$, proved in the following lemma. 
		
		\begin{lemma} \label{lemma}
			The least norm of $\mathbf{H}\mathbf{x}$ for non-zero $\mathbf{x}$ when $\mathbf{x}\in \{0,-1,1\}^{M\times 1}$, equals $1$.
		\end{lemma}
		\begin{proof}
			$\mathbf{H}$ is an invertible matrix. So, for any $\mathbf{x}\neq 0$, $\mathbf{H}\mathbf{x}$ can not be zero. Since all the entries in $\mathbf{H}$ and $\mathbf{x}$ are integers, every entry of  $\mathbf{H}\mathbf{x}$ is also an integer. Therefore, for any non-zero $\mathbf{x}$, $\mathbf{H}\mathbf{x}$ at least has the magnitude $1$.
			 
		\end{proof}
		Note that since $\mathbf{u}-\mathbf{v}\in \{0,2,-2\}^{M\times 1}$, from the above lemma, $\mathbf{H}(\mathbf{u}-\mathbf{v})$ has the minimum magnitude of $2$.
		
		$(c)$ is due to substituting $\sigma$. Therefore
		\begin{align}
		\mathbb{P}(\textrm{error})&\leq 2^M \sum_{\mathbf{u}} \mathbb{P}(\mathbf{x}_n=\mathbf{u}) Q\left(\sqrt{\frac{\alpha (1-2\eta)^2}{4(M+K-1)\eta (1-\eta)}}\right)\\
		&=2^M Q\left(\sqrt{\frac{\alpha (1-2\eta)^2}{4(M+K-1)\eta (1-\eta)}}\right)\\
		&\leq 2^M \exp\left(-\frac{\alpha (1-2\eta)^2}{8(M+K-1)\eta (1-\eta)}\right),
		\end{align}
		where the last inequality follows from $Q(x)\leq\exp(-\frac{x^2}{2}),$ $\forall x\geqslant 0$. For $\mathbb{P}(\textrm{error})$ to be less than $\epsilon$, it is sufficient for the last inequality to be less than $\epsilon$ and then the theorem follows.
		
	\end{proof}
	
	The theorem states that for reconstruction condition to hold, the variance of noise should be kept small enough and this is done by adjusting the value for $\alpha$. Therefore, we can conclude that $\alpha$ is an error-resistant parameter that can be adjusted based on our need.
	
	\subsection{Proof of Theorem \ref{thm2}} \label{proof12}
	Using the mathematical model introduced in~\eqref{sequencer model}, we are ready to provide the proof of Theorem \ref{thm2}.
	\begin{proof}
		The fact is that for each sequencer $s$, $\mathcal{R}^{(s)}$ is equivalent to $q^{(s)}_n$, $\forall n \in [N]$. The reason is that fragments contain just one SNP and are grouped based on their containing SNP position and in the group containing SNP position $n$, the information is stored in $q^{(s)}_n$. Thus we have
		\begin{align}
			\mathbb{P}\left(X_{m,n},m\in\mathcal{A}^{(s)},n\in [N]\ |\ \mathcal{R}^{(s)}\right)&=\prod_{n=1}^{N}\mathbb{P}\left(X_{m,n},m\in\mathcal{A}^{(s)}\ |\ q^{(s)}_n\right).
		\end{align}
		Due to independence of entries in $\mathbf{X}$, we have
		\begin{align}
			\mathbb{P}\left(X_{m,n},m\in\mathcal{A}^{(s)},n\in [N]\right)=\prod_{n=1}^{N}\mathbb{P}\left(X_{m,n},m\in\mathcal{A}^{(s)}\right).
		\end{align}
		Based on the last two equalities
		\begin{align}
			I\left(X_{m,n},m\in\mathcal{A}^{(s)},n\in [N];\mathcal{R}^{(s)}\right)=\sum_{n=1}^{N}I\left(X_{m,n},m\in\mathcal{A}^{(s)};q^{(s)}_n\right).
		\end{align}
		Thus, for privacy condition~\eqref{pri} to be satisfied, using~\eqref{mag}, it is sufficient for every $n\in [N]$ to satisfy
		\begin{align} \label{r3}
		\frac{I\left(X_{m,n},m\in\mathcal{A}^{(s)};q^{(s)}_n\right)}{M-1}\leq \beta.
		\end{align}
		
		Without loss of generality, we have assumed that $s=1$. 
		
		Let
		\begin{align}
		C\left(X_{m,n},m\in\mathcal{A}^{(1)}\right)\triangleq\sum_{m=2}^{M}X_{m,n}+\sum_{k=1}^{K}Y_{k,n}.
		\end{align}
		Based on~\eqref{sequencer model}, we have the following Markov chain  
		\begin{align}
		X_{m,n},m\in\mathcal{A}^{(1)} \rightarrow C\left(X_{m,n},m\in\mathcal{A}^{(1)}\right) \rightarrow q^{(1)}_n.
		\end{align}
		Due to the data processing inequality we have
		\begin{align} \label{r2}
			I\left(X_{m,n},m\in\mathcal{A}^{(1)};q^{(1)}_n\right) \leq I\left(X_{m,n},m\in\mathcal{A}^{(1)};C\left(X_{m,n},m\in\mathcal{A}^{(1)}\right)\right).
		\end{align}
		We can write the right-hand side of the above inequality as
		\begin{align}
		I\left(X_{m,n},m\in\mathcal{A}^{(1)};C\left(X_{m,n},m\in\mathcal{A}^{(1)}\right)\right)=&H\left(C\left(X_{m,n},m\in\mathcal{A}^{(1)}\right)\right)-H\left(C\left(X_{m,n},m\in\mathcal{A}^{(1)}\right)\ |\ X_{m,n},m\in\mathcal{A}^{(1)}\right)\nonumber\\
		=&H\left(C\left(X_{m,n},m\in\mathcal{A}^{(1)}\right)\right)-H\left(\sum_{k=1}^{K}Y_{k,n}\right).
		\end{align}
		Let
		\begin{align}
		M_n\triangleq\sum_{k=1}^{K}Y_{k,n}.
		\end{align}
		Both $C\left(X_{m,n},m\in\mathcal{A}^{(1)}\right)$ and $M_n$ are binomial random variables. In fact
		\begin{align}
		C\left(X_{m,n},m\in\mathcal{A}^{(1)}\right)&\sim\text{Bin}\left(M-1+K,p_n\right)\\
		M_n&\sim\text{Bin}\left(K,p_n\right)
		\end{align}
		The proof is complete.
		
	\end{proof}
	
	\subsection{Discussion} \label{discuss1}
	In theorem above, $K$ is used to guarantee~\eqref{ineqthm3}. Recall that it is the number of added known individuals to the pool and as it becomes larger, the entropies in the numerator of~\eqref{ineqthm3} become closer to each other, so we can choose $K$ such that inequality~\eqref{ineqthm3} is satisfied for any given positive $\beta$.

	The formula for the entropy of binomial distribution is to some degree complicated. To be able to work with~\eqref{ineqthm3}, here, we use some approximations and confirm them by some simulation results. Still, if someone wants to do it exactly, they can use numerical methods to work with~\eqref{ineqthm3}. 
	
	\begin{remark}
		For large values of $K$, based on the central limit theorem, we can use the normal distribution to approximate the entropies in~\eqref{ineqthm3}. Our simulations show that for all values of $K$ (and $M$), the numerator in~\eqref{ineqthm3} is well approximated when we replace the binomial variables with their corresponding normal approximations (see Fig. \ref{fig12}).

		\begin{figure}
			\centering
			\hspace*{-0.7cm}
			\includegraphics[width=0.5\linewidth]{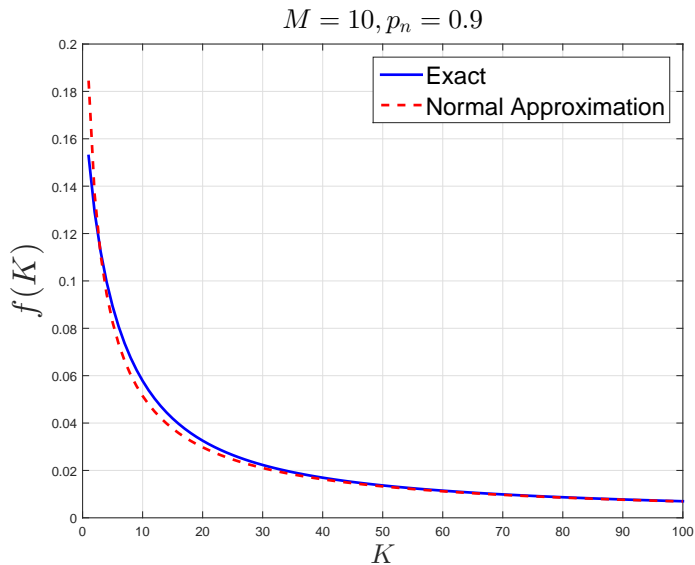}
			\caption{Comparing exact and approximate curves for $f(K)$ in~\eqref{ineqthm3} for different values of $K$ for the case $M=10$ and $p_n=0.9$. As it is seen, the two curves are pretty much the same.}
			\label{fig12}
		\end{figure}
		Approximating the binomial distributions in~\eqref{ineqthm3} by normal distributions, we can rewrite~\eqref{ineqthm3}  as
		\begin{align} \label{qqq}
			f(K)\cong&\frac{1}{2(M-1)}\left(\log\left(2\pi e \left(M-1+K\right)p_n(1-p_n)\right)\right)\nonumber\\
			-&\frac{1}{2(M-1)}\left(\log\left(2\pi e \left(K\right)p_n(1-p_n)\right)\right)\nonumber\\
			=&\frac{1}{2(M-1)}\log\left(\frac{M-1+K}{K}\right).
		\end{align}
	\end{remark}
	
	\begin{remark}
		Based on~\eqref{qqq} and~\eqref{ineqthm3}, we have
		\begin{align} \label{app}
			\frac{M-1+K}{K}\leq\exp\left(\frac{2(M-1)}{\log e}\beta\right).
		\end{align}
		For the case where $\beta\ll 1$, we have
		\begin{align}
			\exp\left(\frac{2(M-1)}{\log e}\beta\right)\cong 1+\frac{2(M-1)}{\log e}\beta.
		\end{align}
		Thus, from~\eqref{app} we have
		\begin{align} \label{MVF}
			K\geq\frac{\log e}{2\beta}.
		\end{align}
		The result is very interesting. It shows that for small values of $\beta$, the number of known individuals needed to preserve the privacy condition is inversely proportional to $\beta$, and has no connection to the number of unknown individuals $M$.
	\end{remark}
	
	\begin{remark}
		Recall that in the all-but-one scheme, $|\mathcal{A}^{(s)}|=M-1$. Here, we argue that this is a good choice. Let us assume that for each sequencer $s\in[S]$ we have
		\begin{align}
			|\mathcal{A}^{(s)}|=c.
		\end{align}
		Following the same math and using the normal approximation, for those cases, to guarantee privacy, we need to have
		\begin{align} \label{53}
			\frac{1}{2c}\log\left(\frac{c+K}{K}\right)\leq\beta.
		\end{align}
		Obviously, increasing $c$ will reduce the left-hand side of~\eqref{53}, so it is better to choose $c$ as large as possible. The maximum value for $c$ is $M$. The problem for $c=M$ is that it makes the reconstruction impossible because in this case, all entries in $\mathbf{H}$ are $1$ and it is not invertible. This justifies the choice of $c=M-1$ which makes not only $\mathbf{H}$ invertible but also helps to satisfy~\eqref{ineqthm3} to its maximum extent. Recall that in the proposed scheme we have two ideas: (1) using multiple sequencers, (2) adding known individuals to the pool. Using known individuals directly helps to preserve the privacy condition. On the other hand, using multiple sequencers is connected to both conditions: (i) it lets using the maximum number of unknown individuals in every sequencer to help preserve privacy, (ii) it satisfies the reconstruction condition by providing a set of equations leading to an invertible $\mathbf{H}$. 
	\end{remark}
	
	\begin{remark}
		Note that using~\eqref{ineqthm3} or approximately~\eqref{MVF}, together with~\eqref{recrec}, it is easy to see that we can always choose $K$ and $\alpha$ such that both conditions are satisfied. First, we choose $K$ such that~\eqref{ineqthm3} is satisfied, then based on that $K$, we use~\eqref{recrec} to achieve $\alpha$.
	\end{remark}
	
	\section{All-But-One Scheme with Random Coverage Depth}\label{sec:random channel}
	
	In the previous section, we assumed that coverage depth is constant for every individual and every SNP position; however in practice, it is not realistic. In this section, we solve the problem in the case of $\alpha^{(s)}_{m,n}$s and $\tilde{\alpha}^{(s)}_{k,n}$s to be random variables which is a more general and practical case. We assume that these variables are binomial random variables and approximate these variables with a normal distribution. Therefore, for $\forall n\in [N], \forall m\in [M], \forall k\in [K], \forall s\in [S], m\neq s$, we have
	\begin{align}
		\alpha^{(s)}_{m,n},\tilde{\alpha}^{(s)}_{k,n}\sim \mathcal{N}\left(\bar{\alpha}, \sigma_{\alpha}^{2}\right).
	\end{align}
	Due to the fact that coverage depths are countable variables and mostly have large values, we assumed that $\bar{\alpha}\in\mathbb{N}$. Also, following the all-but-one scheme, for $\forall n\in [N]$ and $\forall s\in [S]$, we have
	\begin{align}
	\alpha^{(s)}_{s,n}=0.
	\end{align}
	
	Like the previous section, first, we introduce our results. After that, the mathematical model and estimation rule are derived in Subsections \ref{recakhar} and  \ref{est}, respectively. Then, using the model and the estimation rule, the proof is established in Subsection \ref{recakhar}. We discuss the results in Subsection \ref{discuss2}.
	
	The following theorem provides a sufficient condition to satisfy the reconstruction condition.
	\begin{theorem}
		\label{thm3}
		In the all-but-one scheme, the reconstruction constraint~\eqref{rec} is satisfied if:
		\begin{align} \label{eqeq}
		\frac{\bar{\alpha}}{M+K-1} \geq \max\left( \frac{16\eta(1-\eta)}{(1-2\eta)^2}\ln\left(\frac{2^{M+1}}{\epsilon}\right),\right.\nonumber\\
		\left.2\sigma_\alpha\sqrt{\frac{\ln\left(\frac{2^{M+1}}{\epsilon}\right)}{M+K-1}}\right).
		\end{align}
	\end{theorem}
	
	\begin{remark}
		For the privacy condition, Theorem \ref{thm2} is valid here as well. This will be discussed later in Subsection \ref{discuss2}.
	\end{remark}
	
	\subsection{Mathematical Model in Data Collector in the All-But-One Scheme} \label{recakhar}
	For any SNP position $n\in [N]$, the objective for the data collector is to estimate the vector $\mathbf{x}_n=\left[X_{1,n}, X_{2,n},\cdots, X_{M,n}\right]^T$. We define the extended vector $\tilde{\mathbf{x}}_n\triangleq\left[X_{1,n},\cdots,X_{M,n},y_{1,n},\cdots,y_{K,n}\right]^T$, where the last $K$ entries are known to the data collector. Therefore, for the data collector, estimating $\tilde{\mathbf{x}}_n$ is equivalent to estimating $\mathbf{x}_n$.
	
	In the following, we will argue that the knowledge provided by the sequencers to the data collector in SNP position $n$, can be modeled by a set of noisy linear equations as follows	
	\begin{align}\label{random channel}
	\mathbf{g}_n=\tilde{\mathbf{H}}\tilde{\mathbf{x}}_n+\mathbf{z}_n,
	\end{align}
	where $\mathbf{g}_n$ and $\mathbf{z}_n$ are $M\times 1$ vectors, $\mathbf{z}_n\sim\mathcal{N}(0,\sigma^2 \mathbf{I})$, and  
	\begin{align} \label{sigma22}
	\sigma^2\triangleq\frac{4(M+K-1)}{\bar{\alpha}}\frac{\eta(1-\eta)}{(1-2\eta)^2}.
	\end{align}
	Also $\tilde{\mathbf{H}}$ is an $M\times (M+K)$ matrix such that
	\begin{equation} \label{key}
	\tilde{\mathbf{H}}=\mathbf{H}+\Delta,
	\end{equation}
	where
	\begin{align}
	\left\{
	\begin{array}{ll}
	\Delta_{ij} \sim \mathcal{N}\left(0,\,\sigma^{2}_{1}\right), & \quad i\in [M],j\in [M+K]\ \text{and}\ i\neq j \\‎ 
	\Delta_{ii}=0,&\quad \forall i\in [M],
	\end{array}\right.
	\end{align}
and
	\begin{align}
	\sigma_1^2\triangleq\frac{\sigma_\alpha^2}{\bar{\alpha}^2}.
	\end{align}
	Also $\mathbf{H}$ is an $M\times(M+K)$ matrix such that $\mathbf{H}_{ii}=0$ for all $i\in [M]$ and the rest of entries are $1$.
	
	Without loss of generality, let us focus on the first sequencer. As mentioned before, fragments of individuals $2$ to $M$ is sent to sequencer $1$. In pooled sequencing scenario, sequencer $1$ will observe $G^{(1)}_n$, which is defined as 
	\begin{align} \label{mainrandom}
	G^{(1)}_n\triangleq&\frac{1}{\bar{\alpha}(1-2\eta)}\sum_{i=1}^{\alpha^{(1)}_{2,n}} \tilde{X}^{(1)}_{2,n,i} + \cdots +  \frac{1}{\bar{\alpha}(1-2\eta)}\sum_{i=1}^{\alpha^{(1)}_{M,n}} \tilde{X}^{(1)}_{M,n,i} \nonumber\\
	+& \frac{1}{\bar{\alpha}(1-2\eta)}\sum_{i=1}^{\tilde{\alpha}^{(1)}_{1,n}} \tilde{Y}^{(1)}_{1,n,i} + \cdots + \frac{1}{\bar{\alpha}(1-2\eta)}\sum_{i=1}^{\tilde{\alpha}^{(1)}_{K,n}} \tilde{Y}^{(1)}_{K,n,i}.
	\end{align}
	
	Consider the random variable $\frac{1}{\bar{\alpha}}\sum_{i=1}^{\alpha^{(1)}_{2,n}} \tilde{X}^{(1)}_{2,n,i}$ conditioned on $X_{2,n}$. We have
	\begin{align}\label{hey}
	\frac{1}{\bar{\alpha}}\sum_{i=1}^{\alpha^{(1)}_{2,n}} \tilde{X}^{(1)}_{2,n,i}&=\frac{1}{\bar{\alpha}}\sum_{i=1}^{\bar{\alpha}} \tilde{X}^{(1)}_{2,n,i}\nonumber\\
	&+\frac{1}{\bar{\alpha}}\sum_{i=\bar{\alpha}+1}^{\alpha^{(1)}_{2,n}} \tilde{X}^{(1)}_{2,n,i}.
	\end{align}
	It is trivial that the random variables $\frac{1}{\bar{\alpha}}\sum_{i=1}^{\bar{\alpha}} \tilde{X}^{(1)}_{2,n,i}$ and $\frac{1}{\bar{\alpha}}\sum_{i=\bar{\alpha}+1}^{\alpha^{(1)}_{2,n}} \tilde{X}^{(1)}_{2,n,i}$ are independent conditioned on $X_{2,n}$. Also, the distribution of $\frac{1}{\bar{\alpha}}\sum_{i=\bar{\alpha}+1}^{\alpha^{(1)}_{2,n}} \tilde{X}^{(1)}_{2,n,i}$ resembles that of $\frac{1}{\bar{\alpha}}\sum_{i=1}^{\alpha^{(1)}_{2,n}-\bar{\alpha}} \tilde{X}^{(1)}_{2,n,i}$ both conditioned on $X_{2,n}$.

	We define
	\begin{align}
	\xi^{(1)}_{2,n}\triangleq\alpha^{(1)}_{2,n}-\bar{\alpha}.
	\end{align}
	Thus we have $\mathbb{E}\left(\xi^{(1)}_{2,n}\right)=0$ and
	\begin{align}
	\textrm{Var}\left(\xi^{(1)}_{2,n}\right)=\textrm{Var}\left(\alpha^{(1)}_{2,n}\right)=\sigma_\alpha^2.
	\end{align}
	Similar to the steps taken in Section \ref{mathdata} and as a result of the central limit theorem
	\begin{align} \label{rr3}
		\frac{1}{\bar{\alpha}}\sum_{i=1}^{\bar{\alpha}} \tilde{X}^{(1)}_{2,n,i}\ \text{conditioned on}\ X_{2,n}\sim\mathcal{N}\left((1-2\eta)X_{2,n},\frac{4\eta(1-\eta)}{\bar{\alpha}}\right).
	\end{align} 
	
	For the second term in~\eqref{hey} we have
	\begin{align} \label{rr2}
	\mathbb{E}\left(\frac{1}{\bar{\alpha}}\sum_{i=1}^{\xi^{(1)}_{2,n}} \tilde{X}^{(1)}_{2,n,i}\ |\ X_{2,n}\right)&= \mathbb{E}_{\xi^{(1)}_{2,n}}\mathbb{E}\left(\frac{1}{\bar{\alpha}}\sum_{i=1}^{\xi^{(1)}_{2,n}} \tilde{X}^{(1)}_{2,n,i}\ |\ X_{2,n},\xi^{(1)}_{2,n}\right)\nonumber\\
	&=\mathbb{E}_{\xi^{(1)}_{2,n}}\left(\frac{1}{\bar{\alpha}}\xi^{(1)}_{2,n} (1-2\eta)X_{2,n}\right)=0.
	\end{align}
	
	Using the law of total variance we have
	\begin{align} \label{rr1}
	&\textrm{Var}\left(\frac{1}{\bar{\alpha}}\sum_{i=1}^{\xi^{(1)}_{2,n}} \tilde{X}^{(1)}_{2,n,i}\ |\ X_{2,n}\right)\nonumber\\
	&=\mathbb{E}_{\xi^{(1)}_{2,n}}\left(\textrm{Var}\left(\frac{1}{\bar{\alpha}}\sum_{i=1}^{\xi^{(1)}_{2,n}} \tilde{X}^{(1)}_{2,n,i}\ |\ X_{2,n},\xi^{(1)}_{2,n}\right)\right)\nonumber\\
	&+\textrm{Var}_{\xi^{(1)}_{2,n}}\left(\mathbb{E}\left(\frac{1}{\bar{\alpha}}\sum_{i=1}^{\xi^{(1)}_{2,n}} \tilde{X}^{(1)}_{2,n,i}\ |\ X_{2,n},\xi^{(1)}_{2,n}\right)\right)\nonumber\\
	&\stackrel{(a)}{=}\textrm{Var}_{\xi^{(1)}_{2,n}}\left(\frac{1}{\bar{\alpha}}\xi^{(1)}_{2,n} (1-2\eta)X_{2,n}\right)\nonumber\\
	&=\frac{(1-2\eta)^2}{(\bar{\alpha})^2} \textrm{Var}_{\xi^{(1)}_{2,n}}\left(\xi^{(1)}_{2,n}\right)\nonumber\\
	&=\frac{(1-2\eta)^2}{(\bar{\alpha})^2} \sigma_\alpha^2.
	\end{align}
	where $(a)$ results from the fact that $\mathbb{E}\left(\xi^{(1)}_{2,n}\right)=0$. Based on~\eqref{rr1},~\eqref{rr2},~\eqref{rr3},~\eqref{hey} and due to the MMSE rule and the orthogonality theorem we have
	\begin{align} \label{ghar1}
	\frac{1}{\bar{\alpha}(1-2\eta)}\sum_{i=1}^{\alpha^{(1)}_{2,n}} \tilde{X}^{(1)}_{2,n,i}=\frac{\alpha^{(1)}_{2,n}}{\bar{\alpha}}X_{2,n}+Z^{(1)}_{2,n},
	\end{align}
	where $Z^{(1)}_{2,n}\sim\mathcal{N}\left(0,\frac{4\eta(1-\eta)}{\bar{\alpha}(1-2\eta)^2}\right)$. Using the same steps, for the data collector we have
	\begin{align} \label{ghar2}
	\frac{1}{\bar{\alpha}(1-2\eta)}\sum_{i=1}^{\tilde{\alpha}^{(1)}_{2,n}} \tilde{Y}^{(1)}_{1,n,i}=\frac{\tilde{\alpha}^{(1)}_{1,n}}{\bar{\alpha}}y_{1,n}+\tilde{Z}^{(1)}_{1,n},
	\end{align}
	where $\tilde{Z}^{(1)}_{1,n}\sim\mathcal{N}\left(0,\frac{4\eta(1-\eta)}{\bar{\alpha}(1-2\eta)^2}\right)$.
	
	Therefore using~\eqref{ghar1} and~\eqref{ghar2},~\eqref{mainrandom} can be written as
	\begin{align}
	G^{(1)}_n=&\frac{\alpha^{(1)}_{2,n}}{\bar{\alpha}}X_{2,n}+\cdots+\frac{\alpha^{(1)}_{M,n}}
	{\bar{\alpha}}X_{M,n} \nonumber\\
	+&\frac{\tilde{\alpha}^{(1)}_{1,n}}{\bar{\alpha}}y_{1,n}+\cdots+\frac{\tilde{\alpha}^{(1)}_{K,n}}{\bar{\alpha}}y_{K,n}+Z^{(1)}_{n}, \label{mainrand}
	\end{align} 
	where
	\begin{align}
	Z^{(1)}_{n}=\sum_{m=2}^{M}Z^{(1)}_{m,n}+\sum_{k=1}^{K}\tilde{Z}^{(1)}_{k,n}.
	\end{align}
	Thus $Z^{(1)}_{n}\sim\mathcal{N}\left(0,\sigma^2\right)$
	where  $\sigma^2$ follows~\eqref{sigma22}. For the fraction $\frac{\alpha^{(1)}_{m,n}}{\bar{\alpha}}$ we can write it as
	\begin{align} \label{72}
	\frac{\alpha^{(1)}_{m,n}}{\bar{\alpha}}=1+\frac{\zeta^{(1)}_{m,n}}{\bar{\alpha}},
	\end{align}
	where
	\begin{align}
	\zeta^{(1)}_{m,n}\triangleq \alpha^{(1)}_{m,n}-\bar{\alpha}.
	\end{align}
	Therefore $\textrm{Var}(\zeta^{(1)}_{m,n})=\textrm{Var}\left(\alpha^{(1)}_{m,n}\right)$ and
	\begin{align} \label{74}
	\textrm{Var}\left(\frac{\zeta^{(1)}_{m,n}}{\bar{\alpha}}\right)=\frac{\textrm{Var}\left(\alpha^{(1)}_{m,n}\right)}{(\bar{\alpha})^2}=\sigma_1^2.
	\end{align}
	Based on~\eqref{72} and~\eqref{74},~\eqref{key} results.  All $G^{(i)}_n$s are derived in a similar manner. Consequently for $\mathbf{g}_n\triangleq \left[ G_n^{(1)},\cdots,G_n^{(M)}\right]^T$ and $\mathbf{z}_n\triangleq\left[Z_n^{(1)},\cdots,Z_n^{(M)}\right]^T$,~\eqref{random channel} holds.

	\subsection{Estimation Rule} \label{est}
	In this section, our objective is to find the rule that should be used by the data collector to estimate $\tilde{\mathbf{x}}_n$. Using the ML rule, the estimate $\hat{\tilde{\mathbf{x}}}_n$ is obtained by
	\begin{align}
	\hat{\tilde{\mathbf{x}}}_n&=\arg\max_{\tilde{\mathbf{x}}_n} \mathbb{P}(\mathbf{g}_n\ |\ \tilde{\mathbf{x}}_n)=\arg\max_{\tilde{\mathbf{x}}_n} \mathbb{P}(\mathbf{g}_n-\mathbf{H} \tilde{\mathbf{x}}_n\ |\ \tilde{\mathbf{x}}_n) \nonumber \\
	&=\arg\max_{\tilde{\mathbf{x}}_n} \mathbb{P}(\Delta \tilde{\mathbf{x}}_n+\mathbf{z}_n=\mathbf{g}_n-\mathbf{H} \tilde{\mathbf{x}}_n\ |\ \tilde{\mathbf{x}}_n),
	\label{MLeqn}
	\end{align}
	where the last equality is due to~\eqref{random channel} and~\eqref{key}. 
	We should now derive the distribution of $\Delta \tilde{\mathbf{x}}_n+\mathbf{z}_n$ conditioned on $\tilde{\mathbf{x}}_n$. We have
	\begin{align}
	\mathbb{E}(\Delta\tilde{\mathbf{x}}_n+\mathbf{z}_n\ |\ \tilde{\mathbf{x}}_n)=\mathbb{E}(\Delta \tilde{\mathbf{x}}_n\ |\ \tilde{\mathbf{x}}_n)+\mathbb{E}(\mathbf{z}_n\ |\ \tilde{\mathbf{x}}_n).
	\end{align}
	From~\eqref{ghar1}, we have $\mathbb{E}(\mathbf{z}_n\ |\ \tilde{\mathbf{x}}_n)=0$, so
	\begin{align}
	\mathbb{E}(\Delta \tilde{\mathbf{x}}_n+\mathbf{z}_n\ |\ \tilde{\mathbf{x}}_n)=\mathbb{E}(\Delta \tilde{\mathbf{x}}_n\ |\ \tilde{\mathbf{x}}_n)=\mathbb{E}(\Delta\ |\ \tilde{\mathbf{x}}_n)\tilde{\mathbf{x}}_n=0,
	\end{align}
	where the last equality is the result of the independence of $\Delta$ from $\tilde{\mathbf{x}}_n$ and the fact that the expectation of  $\Delta$ is zero.
	For the covariance matrix, we have
	\begin{align}
	&\mathrm{Cov}(\Delta \tilde{\mathbf{x}}_n+\mathbf{z}_n\ |\ \tilde{\mathbf{x}}_n)=\mathbb{E}((\Delta \tilde{\mathbf{x}}_n+\mathbf{z}_n)(\Delta \tilde{\mathbf{x}}_n+\mathbf{z}_n)^T\ |\ \tilde{\mathbf{x}}_n)\nonumber\\
	&=\mathbb{E}((\Delta \tilde{\mathbf{x}}_n+\mathbf{z}_n)(\tilde{\mathbf{x}}_n^T\Delta^T+\mathbf{z}_n^T)\ |\ \tilde{\mathbf{x}}_n)\nonumber\\
	&=\mathbb{E}(\Delta \tilde{\mathbf{x}}_n \tilde{\mathbf{x}}_n^T\Delta^T+\Delta \tilde{\mathbf{x}}_n \mathbf{z}_n^T+\mathbf{z}_n \tilde{\mathbf{x}}_n^T\Delta^T+\mathbf{z}_n\mathbf{z}_n^T\ |\ \tilde{\mathbf{x}}_n)\nonumber\\
	&=\mathbb{E}(\Delta \tilde{\mathbf{x}}_n \tilde{\mathbf{x}}_n^T\Delta^T\ |\ \tilde{\mathbf{x}}_n)+\mathbb{E}(\Delta \tilde{\mathbf{x}}_n \mathbf{z}_n^T\ |\ \tilde{\mathbf{x}}_n)\nonumber\\
	&+\mathbb{E}(\mathbf{z}_n \tilde{\mathbf{x}}_n^T\Delta^T\ |\ \tilde{\mathbf{x}}_n)+\mathbb{E}(\mathbf{z}_n\mathbf{z}_n^T\ |\ \tilde{\mathbf{x}}_n).
	\end{align}
	The second and third terms are zero due to the independence of $\Delta$ and $\tilde{\mathbf{x}}_n$. Also, the entries in $\mathbf{z}_n$ are independent and have the variance $\sigma^2$ (with zero mean), so we have
	\begin{align}
	\mathbb{E}(\mathbf{z}_n\mathbf{z}_n^T\ |\ \tilde{\mathbf{x}}_n)=\sigma^2 \mathbf{I}.
	\end{align}
	Thus,
	\begin{equation}
	\mathrm{Cov}(\Delta \tilde{\mathbf{x}}_n+\mathbf{z}_n\ |\ \tilde{\mathbf{x}}_n)=\mathbb{E}(\Delta \tilde{\mathbf{x}}_n \tilde{\mathbf{x}}_n^T\Delta^T\ |\ \tilde{\mathbf{x}}_n)+\sigma^2 \mathbf{I}.
	\label{eqncov}
	\end{equation}
	We now compute the first term in~\eqref{eqncov}. For $\Delta \tilde{\mathbf{x}}_n$ we have
	\begin{align}
	\Delta \tilde{\mathbf{x}}_n&=
	\begin{bmatrix}
	0&\delta_{12}&\cdots&\delta_{1M}&\cdots&\delta_{1(M+K)}\\
	\delta_{21}&0&\cdots&\delta_{2M}&\cdots&\delta_{2(M+K)}\\
	\vdots\\
	\delta_{M1}&\delta_{M2}&\cdots&0&\cdots&\delta_{M(M+K)}
	\end{bmatrix}
	\begin{bmatrix}
	X_{1,n}\\
	\vdots\\
	X_{M,n}\\
	y_{1,n}\\
	\vdots\\
	y_{K,n}
	\end{bmatrix}\nonumber\\
	&=\begin{bmatrix}
	\sum_{m=1,m\neq 1}^{M}\delta_{1m}X_{m,n}+\sum_{k=1}^{K}\delta_{1(M+k)}y_{k,n}\\
	\sum_{m=1,m\neq 2}^{M}\delta_{2m}X_{m,n}+\sum_{k=1}^{K}\delta_{2(M+k)}y_{k,n}\\
	\vdots\\
	\sum_{m=1,m\neq M}^{M}\delta_{Mm}X_{m,n}+\sum_{k=1}^{K}\delta_{M(M+k)}y_{k,n}\\
	\end{bmatrix}\nonumber\\
	&\triangleq\begin{bmatrix}
	a_1\\
	a_2\\
	\vdots\\
	a_{M}
	\end{bmatrix}\triangleq\mathbf{a}.
	\end{align}
	Therefore for $\Delta \tilde{\mathbf{x}}_n\tilde{\mathbf{x}}_n^T\Delta^T$ we have
	\begin{align}
	\Delta \tilde{\mathbf{x}}_n\tilde{\mathbf{x}}_n^T\Delta^T=\mathbf{a}\mathbf{a}^T
	\end{align}
	
	We have
	\begin{align}
	&\mathbb{E}\left(\left(\mathbf{a}\mathbf{a}^T\right)_{11}\ |\ \tilde{\mathbf{x}}_n\right)=\mathbb{E}\left(a_1^2\ |\ \tilde{\mathbf{x}}_n\right)\nonumber\\
	&=\mathbb{E}\left(\left(	\sum_{m=2}^{M}\delta_{1m}X_{m,n}+\sum_{k=1}^{K}\delta_{1(M+k)}y_{k,n}\right)^2\ |\ \tilde{\mathbf{x}}_n\right)\nonumber\\
	&=\mathbb{E}\left(\sum_{m=2}^{M}\delta^2_{1m}X^{2}_{m,n}+\sum_{k=1}^{K}\delta^2_{1(M+k)}y^{2}_{k,n}\ |\ \tilde{\mathbf{x}}_n\right)\nonumber\\
	&+2\mathbb{E}\left(\sum_{m\neq 1}\sum_{k}\delta_{1m}\delta_{1(M+k)}X_{m,n} y_{k,n}\ |\ \tilde{\mathbf{x}}_n\right)\nonumber\\
	&\stackrel{(a)}{=}\mathbb{E}\left(\sum_{m=2}^{M}\delta^2_{1m}+\sum_{k=1}^{K}\delta^2_{1(M+k)}\right)\nonumber\\
	&+\sum_{m\neq 1}\sum_{k}2X_{m,n} y_{k,n}\mathbb{E}\left(\delta_{1m}\delta_{1(M+k)}\right)\nonumber\\
	&=(M+K-1)\sigma_1^2+0=(M+K-1)\sigma_1^2,
	\end{align}
	where $(a)$ comes from the fact that $X_{m,n}^2=y_{k,n}^2=1$.

	For all  $\mathbb{E}\left((\mathbf{a}\mathbf{a}^T)_{ii}\ |\ \tilde{\mathbf{x}}_n\right),$ $i\in [M]$, the calculations are the same. Thus
	\begin{align}
	\mathbb{E}\left((\mathbf{a}\mathbf{a}^T)_{ii}\ |\ \tilde{\mathbf{x}}_n\right)=(M+K-1)\sigma_1^2, \quad \forall i\in [M].
	\end{align}
	We proceed by calculating this for a non-diagonal entry and we will show that the result is zero and then we generalize our result for all non-diagonal entries. We have
	\begin{align}
	&\mathbb{E}\left(\left(\mathbf{a}\mathbf{a}^T\right)_{12}\ |\ \tilde{\mathbf{x}}_n\right)=\mathbb{E}(a_1a_2\ |\ \tilde{\mathbf{x}}_n)\nonumber\\
	&=\mathbb{E}\left(\left(\sum_{m=1,m\neq 1}^{M}\delta_{1m}X_m+\sum_{k=1}^{K}\delta_{1(M+k)}y_k\right)\nonumber\right.\\
	&\left.\times\left(\sum_{m=1,m\neq 2}^{M}\delta_{2m}X_m+\sum_{k=1}^{K}\delta_{2(M+k)}y_k\right)\ |\ \tilde{\mathbf{x}}_n\right)=0,
	\end{align}
	where the last equality comes from the pairwise independence of any of the two $\delta_{1m}$, $\delta_{1(M+k)}$, $\delta_{2m}$, and $\delta_{2(M+k)}$ variables for all $m\in [M]$ and $k\in [K]$. For non-diagonal entries the calculations are the same. Therefore $\mathbb{E}(\Delta \tilde{\mathbf{x}}_n\tilde{\mathbf{x}}_n^T\Delta^T\ |\ \tilde{\mathbf{x}}_n)$ is zero for non-diagonal entries and $(M+K-1)\sigma_1^2$ for diagonal entries. Thus, based on~\eqref{eqncov} we have
	\begin{align}
	&\mathrm{Cov}(\Delta \tilde{\mathbf{x}}_n+\mathbf{z}_n\ |\ \tilde{\mathbf{x}}_n)
	=\left((M+K-1)\sigma_1^2+\sigma^2\right)\mathbf{I}.
	\end{align}
	Therefore
	\begin{align}
	\Delta\tilde{\mathbf{x}}_n+\mathbf{z}_n\ \text{conditioned on}\ \tilde{\mathbf{x}}_n \sim \mathcal{N}(0,\tilde{\sigma}^2\mathbf{I}),
	\end{align}
	where $\tilde{\sigma}^2\triangleq(M+K-1)\sigma_1^2+\sigma^2$. Therefore, from~\eqref{MLeqn} we have
	\begin{align}
	\hat{\tilde{\mathbf{x}}}_n&=\arg\max_{\tilde{\mathbf{x}}_n} \mathbb{P}(\mathbf{g}_n\ |\ \tilde{\mathbf{x}}_n)=\arg\max_{\tilde{\mathbf{x}}_n}\mathbb{P}(\Delta\tilde{\mathbf{x}}_n+\mathbf{z}_n=\mathbf{g}_n-\mathbf{H}\tilde{\mathbf{x}}_n\ |\ \tilde{\mathbf{x}}_n)\nonumber\\
	&=\arg\max_{\tilde{\mathbf{x}}_n} \frac{1}{(\sqrt{2\pi \tilde{\sigma}^2})^M} \exp \left(\frac{-1}{2\tilde{\sigma}^2}||\mathbf{g}_n-\mathbf{H}\tilde{\mathbf{x}}_n||^2\right)\nonumber\\
	&=\arg\max_{\tilde{\mathbf{x}}_n} \frac{-1}{2\tilde{\sigma}^2}||\mathbf{g}_n-\mathbf{H}\tilde{\mathbf{x}}_n||^2\nonumber\\
	&=\arg\min_{\tilde{\mathbf{x}}_n} ||\mathbf{g}_n-\mathbf{H}\tilde{\mathbf{x}}_n||.
	\end{align}
	Therefore, we should use the MMSE estimation to estimate $\tilde{\mathbf{x}}_n$. 
		
	\subsection{Proof of Theorem \ref{thm3}} \label{proof21}
	We are now ready to provide the proof for Theorem \ref{thm3}.
	
	\begin{proof}	
		Let $\hat{\tilde{\mathbf{x}}}_n$ be the estimation of $\tilde{\mathbf{x}}_n$ by the data collector based on the MMSE estimation.
		\begin{align}
		\label{eq:expand}
		\mathbb{P}(\mathrm{error})=\mathbb{P}(\hat{\tilde{\mathbf{x}}}_n\neq \tilde{\mathbf{x}}_n) = 
		\sum_{\boldsymbol{u}} \mathbb{P}(\tilde{\mathbf{x}}_n=\mathbf{u}) \mathbb{P}(\hat{\tilde{\mathbf{x}}}_n\neq \mathbf{u}\ |\ \tilde{\mathbf{x}}_n=\mathbf{u}).
		\end{align}
		On the other hand, using union bound, we have
		\begin{align} \label{notlast}
		\mathbb{P}(\hat{\tilde{\mathbf{x}}}_n\neq\mathbf{u}\ |\ \tilde{\mathbf{x}}_n=\mathbf{u}) \leq &\sum_{\mathbf{v}\neq \mathbf{u}} \mathbb{P} (\hat{\tilde{\mathbf{x}}}_n= \mathbf{v}\ |\ \tilde{\mathbf{x}}_n=\mathbf{u})\nonumber\\
		&=\sum_{\mathbf{v}\neq \mathbf{u}} \mathbb{P} ( ||\mathbf{g}_n-\mathbf{H}\mathbf{u}||_2 \geq ||\mathbf{g}_n-\mathbf{H}\mathbf{v}||_2 \ |\ \tilde{\mathbf{x}}_n=\mathbf{u})\nonumber\\
		&\stackrel{(a)}{=}\sum_{\mathbf{v}\neq \mathbf{u}} \Pr(||\Delta \mathbf{u}+\mathbf{z}_n||_2\geq\\
		&\qquad||\Delta \mathbf{u}+\mathbf{z}_n+\mathbf{H}(\mathbf{u}-\mathbf{v})||_2\ |\ \tilde{\mathbf{x}}_n=\mathbf{u})\nonumber\\
		&\stackrel{(b)}{=}\sum_{\mathbf{v}\neq \mathbf{u}} \mathbb{P}\left(-2(\Delta \mathbf{u}+\mathbf{z}_n)\cdot(\mathbf{H}(\mathbf{u}-\mathbf{v}))\right.\geq\\
		&\qquad\left.||\mathbf{H}(\mathbf{u}-\mathbf{v})||_2^2\ |\ \tilde{\mathbf{x}}_n=\mathbf{u}\right)\nonumber\\
		&=\sum_{\mathbf{v}\neq \mathbf{u}} \mathbb{P}\left( -2\mathbf{z}_n\cdot(\mathbf{H}(\mathbf{u}-\mathbf{v}))\geq\right.\\
		&\left.\qquad||\mathbf{H}(\mathbf{u}-\mathbf{v})||_2^2+2\Delta \mathbf{u}\cdot(\mathbf{H}(\mathbf{u}-\mathbf{v}))\ |\ \tilde{\mathbf{x}}_n=\mathbf{u}\right)\nonumber\\
		&=\sum_{\mathbf{v}\neq \mathbf{u}} \mathbb{P}\left(\frac{-\mathbf{z}_n\cdot(\mathbf{H}(\mathbf{u}-\mathbf{v}))}{||\mathbf{H}(\mathbf{u}-\mathbf{v})||_2\sigma}\geq\right.\\
		&\qquad\left.\frac{||\mathbf{H}(\mathbf{u}-\mathbf{v})||_2}{2\sigma}+\left(\frac{\Delta \mathbf{u}}{\sigma}\right)\cdot\frac{\mathbf{H}(\mathbf{u}-\mathbf{v})}{||\mathbf{H}(\mathbf{u}-\mathbf{v})||_2}\ |\ \tilde{\mathbf{x}}_n=\mathbf{u}\right)\nonumber\\
		\label{notlast2}
		&\stackrel{(c)}{=}\sum_{\mathbf{v}\neq \mathbf{u}} \mathbb{E}_\Delta\left(Q\left(\frac{||\mathbf{H}(\mathbf{u}-\mathbf{v})||_2}{2\sigma}+\left(\frac{\Delta \mathbf{u}}{\sigma}\right)\cdot\frac{\mathbf{H}(\mathbf{u}-\mathbf{v})}{||\mathbf{H}(\mathbf{u}-\mathbf{v})||_2}\right)\right)\nonumber\\
		&\stackrel{(d)}{\leq}\sum_{\mathbf{v}\neq \mathbf{u}} \int_{-\infty}^{\infty}Q\left(\frac{1}{\sigma}+x\right)\frac{1}{\sqrt{2\pi \sigma_X^2}}\exp\left(\frac{-x^2}{2\sigma_X^2}\right) dx,
		\end{align}
where 
		\begin{align}
		\sigma_X^2\triangleq(M+K-1)\left(\frac{\sigma_1}{\sigma}\right)^2.
		\end{align}
		Also, $(a)$ follows from~\eqref{random channel}, $(b)$ follows from squaring the two sides of inequality in $(a)$, and $(c)$ follows from the fact that
		\begin{align}
		\frac{-\mathbf{z}_n\cdot(\mathbf{H}(\mathbf{u}-\mathbf{v}))}{||\mathbf{H}(\mathbf{u}-\mathbf{v})||\sigma}\sim\mathcal{N}(0,1).
		\end{align}
		Inequality ${(d)}$ comes from the fact that
		\begin{align}
		\left(\frac{\Delta \mathbf{u}}{\sigma}\right)\cdot\frac{\mathbf{H}(\mathbf{u}-\mathbf{v})}{||\mathbf{H}(\mathbf{u}-\mathbf{v})||} \sim \mathcal{N}\left(0,\sigma_X^2\right),
		\end{align}
		and that the minimum magnitude for $\mathbf{H}(\mathbf{u}-\mathbf{v})$ is $2$ with a similar reasoning as in Lemma \ref{lemma}. Keep in mind that the last $K$ entries in $\mathbf{u}$ and $\mathbf{v}$ are known a--priori to the data collector and are constant, not variables.
		
		The right-hand side of~\eqref{notlast} can be upper-bounded as, 
		
				\begin{align}
		&\stackrel{(e)}{\leq}\sum_{\mathbf{v}\neq \mathbf{u}}\left(\int_{-\infty}^{\frac{-1}{\sigma}}\frac{1}{\sqrt{2\pi \sigma_X^2}}\exp\left(\frac{-x^2}{2\sigma_X^2}\right) dx+\nonumber\right.\\
		&\left.\int_{\frac{-1}{\sigma}}^{\infty}\exp\left(\frac{-1}{2}\left(\frac{1}{\sigma}+x\right)^2\right)\frac{1}{\sqrt{2\pi \sigma_X^2}}\exp\left(\frac{-x^2}{2\sigma_X^2}\right) dx\right)\nonumber\\
		&\stackrel{(f)}{\leq}\sum_{\mathbf{v}\neq \mathbf{u}} \left(Q\left(\frac{1}{\sigma}\right)+\nonumber\right.\\
		&\left.\frac{1}{\sqrt{1+(M+K-1)(\frac{\sigma_1}{\sigma})^2}}\exp\left(\frac{-1}{2\sigma^2+2(M+K-1)\sigma_1^2}\right)\nonumber\right)\\
		&= 2^M Q\left(\frac{1}{\sigma}\right)+\nonumber\\
		&\frac{2^M}{\sqrt{1+(M+K-1)(\frac{\sigma_1}{\sigma})^2}}\exp\left(\frac{-1}{2\sigma^2+2(M+K-1)\sigma_1^2}\right),
		\end{align} 
		where $(e)$ follows from 
		\begin{align}
		\label{ineq:Q(x)}
		Q(x)\leq
		\begin{cases}
		1, &x\leq 0\\
		\exp(\frac{-x^2}{2}),&x\geq 0.
		\end{cases}
		\end{align}
		In addition, equality $(f)$ is based on the  following upper bound 
		\begin{align}
		&\frac{1}{\sqrt{2\pi \sigma_X^2}}\int_{\frac{-1}{\sigma}}^{\infty}\exp\left(\frac{-1}{2}\left(\frac{1}{\sigma}+x\right)^2+\frac{-x^2}{2\sigma_X^2}\right)dx\nonumber\\
		&=\frac{1}{\sqrt{2\pi \sigma_X^2}}\nonumber\\
		&\int_{\frac{-1}{\sigma}}^{\infty}\exp\left(\frac{-1}{2\sigma^2}\frac{1}{1+\sigma_X^2}-\left(\sqrt{1+\frac{1}{\sigma_X^2}}x+\frac{1}{\sigma\sqrt{1+\frac{1}{\sigma_X^2}}}\right)^2\right)dx\nonumber\\
		&\leq\frac{1}{\sqrt{2\pi \sigma_X^2}}\exp\left(\frac{-1}{2\sigma^2}\frac{1}{1+\sigma_X^2}\right)\sqrt{2\pi\frac{\sigma_X^2}{1+\sigma_X^2}}\nonumber\\ \label{ineq:lastlast}
		&=\frac{1}{\sqrt{1+(M+K-1)(\frac{\sigma_1}{\sigma})^2}}\exp\left(\frac{-1}{2\sigma^2+2(M+K-1)\sigma_1^2}\right).
		\end{align}
		
		Therefore, from~\eqref{eq:expand} and~\eqref{ineq:lastlast} we have
		\begin{align} \label{er}
		&\mathbb{P}(\textrm{error})\leq 2^M Q\left(\frac{1}{\sigma}
		\right)\nonumber\\
		&+ \frac{2^M}{\sqrt{1+(M+K-1)(\frac{\sigma_1}{\sigma})^2}}\exp\left(\frac{-1}{2\sigma^2+2(M+K-1)\sigma_1^2}\right).
		\end{align}
		For the probability of error to be less than $\epsilon$, it is sufficient for both terms in the right-hand side of the inequality above to be less than $\frac{\epsilon}{2}$. For the first term, using~\eqref{ineq:Q(x)}, we have
		\begin{align}
		2^M Q\left(\frac{1}{\sigma}\right) \leq 2^M\exp\left(-\frac{1}{2}\left(\frac{1}{\sigma}\right)^2\right).
		\end{align}
		It is sufficient to put the right-hand side less than $\frac{\epsilon}{2}$. Thus
		\begin{align}
		\sigma^2 \leq \frac{1}{2\ln\left(\frac{2^{M+1}}{\epsilon}\right)}.
		\label{smallsigma}
		\end{align}
		Substituting $\sigma^2=	\frac{4(M+K-1)}{\bar{\alpha}}\frac{\eta (1-\eta)}{(1-2\eta)^2}$ in the above inequality, we have 
		\begin{align} \label{ineq1}
		\frac{\bar{\alpha}}{M+K-1} \geq \frac{8\eta (1-\eta)}{(1-2\eta)^2}\ln\left(\frac{2^{M+1}}{\epsilon}\right).
		\end{align}
		For the second term in the right-hand side of~\eqref{er}, we have
		\begin{align} \label{label}
		&\frac{2^M}{\sqrt{1+(M+K-1)(\frac{\sigma_1}{\sigma})^2}}\exp\left(\frac{-1}{2\sigma^2+2(M+K-1)\sigma_1^2}\right)\nonumber\\
		&\leq 2^M \exp\left(\frac{-1}{2\sigma^2+2(M+K-1)\sigma_1^2}\right).
		\end{align}
		By putting the right-hand side less than $\frac{\epsilon}{2}$ we have
		\begin{align} \label{ineq3}
		2\sigma^2+2(M+K-1)\sigma_1^2 \leq \frac{1}{\ln\left(\frac{2^{M+1}}{\epsilon}\right)}.
		\end{align}
		In order the inequality above to hold, it is sufficient for both terms in the left-hand side of it to be less than a half of the right-hand side; i.e.
		\begin{align}\label{n1}
		2\sigma^2\leq\frac{1}{2\ln\left(\frac{2^{M+1}}{\epsilon}\right)},
		\end{align}
		and
		\begin{align}\label{n2}
		2(M+K-1)\sigma_1^2 \leq \frac{1}{2\ln\left(\frac{2^{M+1}}{\epsilon}\right)}.
		\end{align}
		Substituting $\sigma^2=	\frac{4(M+K-1)}{\bar{\alpha}}\frac{\eta (1-\eta)}{(1-2\eta)^2}$ in~\eqref{n1} results in
		\begin{align} \label{ineq2}
		\frac{\bar{\alpha}}{M+K-1} \geq \frac{16\eta(1-\eta)}{(1-2\eta)^2}\ln\left(\frac{2^{M+1}}{\epsilon}\right).
		\end{align}
		Substituting $\sigma_1^2=\frac{\sigma_\alpha^2}{\bar{\alpha}^2}$ in~\eqref{n2} results in
		\begin{align}
		\bar{\alpha}\geq 2\sigma_\alpha\sqrt{(M+K-1)\ln\left(\frac{2^{M+1}}{\epsilon}\right)},
		\end{align}
		or equivalently
		\begin{align} \label{ineq4}
		\frac{\bar{\alpha}}{M+K-1}\geq 2\sigma_\alpha\sqrt{\frac{\ln\left(\frac{2^{M+1}}{\epsilon}\right)}{M+K-1}}.
		\end{align}
		Therefore, inequalities~\eqref{ineq1},~\eqref{ineq2}, and~\eqref{ineq4}, together,  provide a sufficient condition to satisfy the reconstruction condition. We note that~\eqref{ineq1} is dominated by~\eqref{ineq2}, and thus it is enough to consider~\eqref{ineq2} and~\eqref{ineq4}.

	\end{proof}
	
	\subsection{Discussions} \label{discuss2}
	\begin{remark}
	In this remark, we argue why Theorem~\ref{thm2} provides a sufficient condition for the all-but-one scheme not only with constant depth of coverage but also with random depth of coverage.

		Following the same steps as in Subsection \ref{recakhar}, we can see that sequencer $1$ exctracts $q^{(1)}_n$ in SNP position $n\in [N]$, where 
				\begin{align}
		q^{(1)}_n=&\frac{\alpha^{(1)}_{2,n}}{\bar{\alpha}}X_{2,n}+\cdots+\frac{\alpha^{(1)}_{M,n}}{\bar{\alpha}}X_{M,n}\nonumber\\
		+&\frac{\tilde{\alpha}^{(1)}_{1,n}}{\bar{\alpha}}Y_{1,n}+\cdots+\frac{\tilde{\alpha}_{K,n}}{\bar{\alpha}}Y_{K,n}+\tilde{Z}^{(1)}_n,
		\end{align}
		where $\textrm{Var}(\tilde{Z}^{(1)}_n)=\sigma^2$ and  $\sigma^2$ is defined in~\eqref{sigma22}. Similar to Subsection \ref{recakhar}, we can rewrite the equation above as:
		\begin{align}
		q^{(1)}_n=&X_{2,n}+\cdots+X_{M,n}\nonumber\\
		+&Y_{1,n}+\cdots+Y_{K,n}+\hat{Z}^{(1)}_n,
		\end{align}
		where $\textrm{Var}\left(\hat{Z}^{(1)}_n\right)=(M+K-1)\sigma_1^2+\sigma^2$. Due to the equality above, the Markov chain used in the proof of Theorem \ref{thm2}, is valid here too. Then taking the same steps as in that proof, we conclude that Theorem \ref{thm2} works here as well, and can be used to adjust $K$. In addition,  all discussions in Subsection \ref{discuss1} holds here as well. 

	\end{remark}
	
	\begin{remark}
	
	For reconstruction, we should keep in mind that there are two sources of noise in~\eqref{random channel}; one results from the randomness in the values of coverage depth and the other results from the sequencers' error in the reading phase, which have the variances $(M+K-1)\sigma_1^2$ and $\sigma^2$ respectively. Both of these variances are functions of $\bar{\alpha}$;  $\sigma_1^2$ is inversely proportional to $\bar{\alpha}^2$ and $\sigma^2$ is inversely proportional to $\bar{\alpha}$, so by increasing $\bar{\alpha}$, both variances become smaller, and can be adjusted to satisfy the reconstruction condition. First, from the privacy condition, we set $K$ based on Theorem \ref{thm2}, and then using the result,  we set $\bar{\alpha}$ from the reconstruction condition based on Theorem \ref{thm3}.  Keep in mind that if $\sigma_\alpha$ is zero, we reach a result similar to that of Theorem \ref{thm1} which was expected.
\end{remark}

	\begin{remark}
		Till now, all the paper was about haploid cells. In the case of diploid cells, there are two complete sets of chromosomes, in contrast to haploid cells in which there is one set \cite{HapDip}. In fact, in diploid cells, in each SNP position, there exist two SNPs, one for each chromosome containing that position. Therefore, in that case, we can assume that each individual represents two haploid-celled individuals--one set of chromosomes for each hypothetic individual. Thus, all theorems established for haploid cells can be used for diploid cells with some slight modifications of replacing $M$ and $K$ with $2M$ and $2K$, respectively.  
	\end{remark}

\section{Conclusion and Future Steps} \label{conclusion}
In this paper, we introduced the problem of privacy in the process of DNA sequencing, where the objective is to keep the DNA information private from the sequencer itself and at the same time, allow us to reconstruct the sequence(s) in a local trusted processor. This is in contrast with the previous results, where leaking information to the sequencer itself was accepted and the focus was on keeping this data private later where it is processed or stored.

To satisfy privacy condition at the sequencer(s) and reconstruction condition at the data collector, we develop an information gap between the sequencer(s) and the data collector using two techniques: (1) using more than one non-colluding sequencers, which all send the results of the read fragments to the single data collector, (2) adding the fragments of some known sequences to the pool of DNA fragments. These fragments are known to the data collector, but unknown to the sequencers. These two techniques provide enough freedom to satisfy both privacy and reconstruction conditions at the same time.

This work is the first step toward this direction, and there are many open challenges to be addressed in future, for instance, the case in which a subset of sequencers are colluding could be concerned. Moreover, considering the case in which fragments are not limited to contain just one SNP is of interest.  
Also, the lower bounds in the theorems in this paper can be improved.

	\bibliographystyle{ieeetr}
	\bibliography{journal_abbr,CodDecB}

\begin{thebibliography}{10}

\bibitem{lander2011initial}
E.~S. Lander, ``Initial impact of the sequencing of the human genome,'' {\em
  Nature}, vol.~470, no.~7333, p.~187, 2011.

\bibitem{Stan}
S.~Medicine, {\em Forecasting risk of deadly vascular condition from genome
  sequence}, 2018.
\newblock \url{https://www.sciencedaily.com/releases/2018/09/180906141634.htm}.

\bibitem{lander1996new}
E.~S. Lander, ``The new genomics: global views of biology,'' {\em Science},
  vol.~274, no.~5287, pp.~536--539, 1996.

\bibitem{risch1996future}
N.~Risch and K.~Merikangas, ``The future of genetic studies of complex human
  diseases,'' {\em Science}, vol.~273, no.~5281, pp.~1516--1517, 1996.

\bibitem{reich2001allelic}
D.~E. Reich and E.~S. Lander, ``On the allelic spectrum of human disease,''
  {\em TRENDS in Genetics}, vol.~17, no.~9, pp.~502--510, 2001.

\bibitem{grosse2006clinical}
S.~D. Grosse and M.~J. Khoury, ``What is the clinical utility of genetic
  testing?,'' {\em Genetics in Medicine}, vol.~8, no.~7, p.~448, 2006.

\bibitem{american2003american}
A.~S. of~Clinical~Oncology {\em et~al.}, ``American society of clinical
  oncology policy statement update: genetic testing for cancer
  susceptibility,'' {\em Journal of clinical oncology: official journal of the
  American Society of Clinical Oncology}, vol.~21, no.~12, p.~2397, 2003.

\bibitem{GenomeTheft}
M.~White, {\em Why you should be scared of someone stealing your genome}, 2013.
\newblock
  \url{https://psmag.com/environment/why-you-should-be-scared-of-someone-stealing-your-genome-58082}.

\bibitem{heeney2011assessing}
C.~Heeney, N.~Hawkins, J.~de~Vries, P.~Boddington, and J.~Kaye, ``Assessing the
  privacy risks of data sharing in genomics,'' {\em Public health genomics},
  vol.~14, no.~1, pp.~17--25, 2011.

\bibitem{knoppers2002genetic}
B.~M. Knoppers, ``Genetic information and the family: are we our brother's
  keeper?,'' {\em Trends in biotechnology}, vol.~20, no.~2, pp.~85--86, 2002.

\bibitem{kaye2012tension}
J.~Kaye, ``The tension between data sharing and the protection of privacy in
  genomics research,'' {\em Annual review of genomics and human genetics},
  vol.~13, pp.~415--431, 2012.

\bibitem{lunshof2008genetic}
J.~E. Lunshof, R.~Chadwick, D.~B. Vorhaus, and G.~M. Church, ``From genetic
  privacy to open consent,'' {\em Nature Reviews Genetics}, vol.~9, no.~5,
  p.~406, 2008.

\bibitem{malin2004not}
B.~Malin and L.~Sweeney, ``How (not) to protect genomic data privacy in a
  distributed network: using trail re-identification to evaluate and design
  anonymity protection systems,'' {\em Journal of biomedical informatics},
  vol.~37, no.~3, pp.~179--192, 2004.

\bibitem{lowrance2007identifiability}
W.~W. Lowrance and F.~S. Collins, ``Identifiability in genomic research,'' {\em
  Science}, vol.~317, no.~5838, pp.~600--602, 2007.

\bibitem{sweeney2002k}
L.~Sweeney, ``k-anonymity: A model for protecting privacy,'' {\em International
  Journal of Uncertainty, Fuzziness and Knowledge-Based Systems}, vol.~10,
  no.~05, pp.~557--570, 2002.

\bibitem{aggarwal2005k}
C.~C. Aggarwal, ``On k-anonymity and the curse of dimensionality,'' in {\em
  Proceedings of the 31st international conference on Very large data bases},
  pp.~901--909, VLDB Endowment, 2005.

\bibitem{johnson2013privacy}
A.~Johnson and V.~Shmatikov, ``Privacy-preserving data exploration in
  genome-wide association studies,'' in {\em Proceedings of the 19th ACM SIGKDD
  international conference on Knowledge discovery and data mining},
  pp.~1079--1087, ACM, 2013.

\bibitem{fienberg2011privacy}
S.~E. Fienberg, A.~Slavkovic, and C.~Uhler, ``Privacy preserving gwas data
  sharing,'' in {\em Data Mining Workshops (ICDMW), 2011 IEEE 11th
  International Conference on}, pp.~628--635, IEEE, 2011.

\bibitem{kaye2009data}
J.~Kaye, C.~Heeney, N.~Hawkins, J.~De~Vries, and P.~Boddington, ``Data sharing
  in genomics--re-shaping scientific practice,'' {\em Nature Reviews Genetics},
  vol.~10, no.~5, p.~331, 2009.

\bibitem{erlich2014routes}
Y.~Erlich and A.~Narayanan, ``Routes for breaching and protecting genetic
  privacy,'' {\em Nature Reviews Genetics}, vol.~15, no.~6, pp.~409--421, 2014.

\bibitem{chor1995private}
B.~Chor, O.~Goldreich, E.~Kushilevitz, and M.~Sudan, ``Private information
  retrieval,'' in {\em Foundations of Computer Science, 1995. Proceedings.,
  36th Annual Symposium on}, pp.~41--50, IEEE, 1995.

\bibitem{sun2017capacity}
H.~Sun and S.~A. Jafar, ``The capacity of private information retrieval,'' {\em
  IEEE Transactions on Information Theory}, vol.~63, no.~7, pp.~4075--4088,
  2017.

\bibitem{banawan2018capacity}
K.~Banawan and S.~Ulukus, ``The capacity of private information retrieval from
  coded databases,'' {\em IEEE Transactions on Information Theory}, vol.~64,
  no.~3, pp.~1945--1956, 2018.

\bibitem{gertner2000protecting}
Y.~Gertner, Y.~Ishai, E.~Kushilevitz, and T.~Malkin, ``Protecting data privacy
  in private information retrieval schemes,'' {\em Journal of Computer and
  System Sciences}, vol.~60, no.~3, pp.~592--629, 2000.

\bibitem{shen2013comprehensive}
H.~Shen, J.~Li, J.~Zhang, C.~Xu, Y.~Jiang, Z.~Wu, F.~Zhao, L.~Liao, J.~Chen,
  Y.~Lin, {\em et~al.}, ``Comprehensive characterization of human genome
  variation by high coverage whole-genome sequencing of forty four
  caucasians,'' {\em PLoS One}, vol.~8, no.~4, p.~e59494, 2013.

\bibitem{langmead2009ultrafast}
B.~Langmead, C.~Trapnell, M.~Pop, and S.~L. Salzberg, ``Ultrafast and
  memory-efficient alignment of short dna sequences to the human genome,'' {\em
  Genome biology}, vol.~10, no.~3, p.~R25, 2009.

\bibitem{HapDip}
A.~Klazema, {\em Haploid vs Diploid Cells: How to Know the Difference}, 2014.
\newblock \url{https://blog.udemy.com/haploid-vs-diploid/}.

\end{thebibliography}
	
	\newpage
	
	\vfill

\end{document}